%% file: main.tex
\documentclass[a4paper,11pt]{article}

\usepackage{fullpage}
\usepackage{times}
\usepackage{soul}
\usepackage{url}
\usepackage[utf8]{inputenc}
\usepackage[small]{caption}
\usepackage{graphicx}
\usepackage[table]{xcolor}
\usepackage{amsmath}
\usepackage{booktabs}
\usepackage{algorithm}
\usepackage{dsfont}
\usepackage{enumerate}
\usepackage{tikz}

\usepackage{amsthm}
\usepackage{amssymb}
\usepackage{algorithm}
\usepackage{algorithmic}
\usepackage{bbm}
\usepackage[T1]{fontenc}

\usepackage{natbib}
\usepackage{amsmath}
\usepackage{mathtools}
\usepackage{nicefrac}
\usepackage{tcolorbox}
\usepackage{thmtools}
\usepackage{thm-restate}
\usepackage{tikz}
\usepackage[capitalize,noabbrev]{cleveref}
\input{macros.tex}

\allowdisplaybreaks

\usepackage{authblk}

\title{\bf Fairness in Repeated Matching: A Maximin Perspective}

\author[1]{Eugene Lim}
\author[2]{Tzeh Yuan Neoh}
\author[3]{Nicholas Teh}

\affil[1]{National University of Singapore, Singapore}
\affil[2]{Harvard University, USA}
\affil[3]{University of Oxford, UK}

\date{\vspace{-10mm}}
\begin{document}

\maketitle
\begin{abstract}
    We study a sequential decision-making model where a set of items is repeatedly matched to the same set of agents over multiple rounds. The objective is to determine a sequence of matchings that either maximizes the utility of the least advantaged agent at the end of all rounds (optimal) or at the end of every individual round (anytime optimal). We investigate the computational challenges associated with finding (anytime) optimal outcomes and demonstrate that these problems are generally computationally intractable. However, we provide approximation algorithms, fixed-parameter tractable algorithms, and identify several special cases whereby the problem(s) can be solved efficiently. Along the way, we also establish characterizations of Pareto-optimal/maximum matchings, which may be of independent interest to works in matching theory and house allocation.
\end{abstract}

\section{Introduction}
Traditional machine learning (ML) algorithms often focus on global objectives such as efficiency (e.g., maximizing accuracy or minimizing error rates in decision-making systems) or maximizing revenue/profit (e.g., maximizing click-through rates for recommendation systems), as they align closely with organizational goals and are more straightforward to quantify and optimize. However, modern approaches increasingly emphasize \emph{fairness} as a key desideratum, as societal and regulatory demands push for more equitable and responsible ML systems.

We consider a multi-agent sequential decision-making scenario where a set of resources must be allocated among agents repeatedly over time, with the objective of achieving fairness in the assignment process. This framework encompasses applications such as dynamic spectrum allocation in wireless networks and energy distribution in smart grids \citep{elhachmi2022spectrum,jain2022energyscheduling,rony2021dynamicspectrum,soares2024localenergysystems}. In the case of spectrum allocation, communication channels must be repeatedly assigned to devices, with each device requiring exclusive access to one channel in each time slot. Persistent disparities in access can degrade system efficiency, reduce user satisfaction, and undermine trust. Similarly, in many other ML-driven resource allocation systems, disparities in the distribution of resources---such as GPUs in distributed computing---can lead to unfair outcomes that compromise the perceived and actual effectiveness of the system. Numerous other applications where decisions are made dynamically---such as assigning tasks to workers in crowdsourcing platforms \citep{moayedikia2020taskassignment}, or distributing compute resources in cloud systems \citep{belgacem2022dynamiccloud,gupta2017dynamiccloud,saraswathi2015dynamiccloud}---call for central decision-makers to ensure that no agent is persistently disadvantaged, which is critical for both fairness and long-term trust in the system.

The scenarios described above can be captured using the \emph{repeated matching} framework---a multi-agent sequential decision-making model in which a set of goods is repeatedly matched to agents over time, and each agent is assigned exactly one good at each round. This can also be viewed as a multi-round generalization of the \emph{bottleneck assignment problem} \citep{ford1962flows} which is well-known in multi-agent task allocation: an application of this problem arises in \emph{threat seduction}, where decoys are assigned to multiple incoming threats \citep{shames2017decoy}. Our problem can also be viewed as a sequential variant of the \emph{Santa Claus problem} \citep{bansal2006santaclaus}, which is closely related to the classic scheduling problem of \emph{makespan minimization} on unrelated parallel machines \citep{lenstra1990schedulingunrelated,santaclausmakespan2024}.

In particular, we focus on the \emph{maximin} (or \emph{egalitarian}) objective \citep{demko1988egalitarian,thomson1983egalitarian}, which aims to find a sequence of matchings that maximizes the minimum utility among agents. Maximin fairness serves as a principled trade-off between fairness and efficiency, as minimizing disparities often enhances overall system robustness and user satisfaction. Moreover, modern ML systems often involve iterative, data-driven decision-making, and maximin fairness integrates naturally with these systems by providing a fairness criterion that adapts dynamically, with its ability to handle both short-term and long-term outcomes.\footnote{This is in contrast to other \emph{comparative} notions of fairness, such as \emph{envy-freeness}, which has also been studied in the static matching \citep{aigner2022envyfreematchings,wu2018envyfreematchings,yokoi2020envyfreematchings} and the two-sided repeated matching \citep{gollapudi2020matching} setting. Maximin fairness is also more \emph{demonstrably} fair compared to an envy-based approach.}

\subsection{Our Results}
 We study the \emph{repeated matching} problem from the perspective of \emph{maximin} (or \emph{egalitarian}) fairness, a principle grounded in game theory, fair division, and matching problems. Leveraging techniques from classical matching algorithms, approximation methods, dynamic programming, and online decision-making, we analyze how to design fair repeated matching policies that ensure long-term fairness across multiple rounds.

 In \Cref{sec:prelims}, we formally define the repeated matching problem and introduce the notion of \emph{(anytime) optimality} in the egalitarian sense. We also introduce several tools that is central in proving some of our results.

In \Cref{sec:optimal-hardness}, we study the computation of optimal solutions. We begin by defining the decision variant of our matching problem and showing that it is NP-hard in general. Notably, this hardness holds even with only two timesteps and ternary agent valuations (i.e., when each agent's utility for a good takes one of three possible values). Given these hardness results, we turn to the optimization variant of the problem and develop approximation algorithms that achieve an additive approximation bound independent of the number of rounds~$T$. Crucially, this implies that as $T$ increases---a scenario common in real-world applications---the solution produced by our algorithm converges to the optimal one. In addition, we also show that the problem is fixed-parameter tractable (FPT) with respect to the number of agents by providing a polynomial-time algorithm when the number of agents is a constant. Notably, in the process, we derive a characterization of \emph{Pareto optimal} matchings in terms of the permutations of agents. This generalizes the previously-known result that \emph{serial dictatorship} characterizes Pareto optimal matchings and may be of independent interest to communities working on the \emph{house allocation} problem.

In \Cref{sec:anytime-optimality}, we shift our focus to \emph{anytime optimal} solutions. We show that such solutions always exist for two agents, and we provide a polynomial-time algorithm for it. However, this does not extend to three or more agents---even with just two rounds, deciding if an instance admits an anytime optimal solution becomes coNP-hard. Nevertheless, we design an approximation algorithm that achieves anytime optimality with an additive bound independent of~$T$. These results underscore the inherent difficulty of achieving anytime optimality in our setting.

In \Cref{sec:special-cases}, we revisit optimality and identify three special cases admitting polynomial-time algorithms: (i) agents with binary valuations, (ii) two types of goods, and (iii) identical agent valuations. These special cases are well-motivated by the (temporal) fair division literature. For (i), we present an exact algorithm and a new characterization of Pareto optimal matchings under binary valuations. For (ii), we similarly provide an efficient exact algorithm. For (iii), despite NP-hardness in general, we show that optimal solutions can be computed in polynomial time when the number of rounds is a multiple of the number of agents. Finally, we extend our approximation approach to anytime optimality in these cases, giving us a stronger results than in the general setting.

\subsection{Related Work}
We highlight several streams of research that are related to our work. We note that while there are many works on \emph{online} matching and fair division, they are not directly relevant to our setting, as the underlying assumptions differ fundamentally. In our setting, the entire set of goods is made available in every round, whereas in online models, the set of goods may vary over time. Thus, we focus only on discussing works where meaningful implications can be drawn between their results and ours.

\paragraph{Repeated matching.}
Repeated matching was first studied by \citet{hosseini2015matchingordinal}, which considered ordinal preferences that could change over time. They study strategyproofness and approximate envy-freeness.
However, ordinal (their model) and cardinal (our model) preferences are vastly different, both in techniques and results. \citet{gollapudi2020matching} subsequently looked at a two-sided repeated matching problem (i.e., each side have preferences over the other side). They also study approximate envy-freeness as the key desiderata, albeit under some strong assumptions. In contrast, our model is on one-sided repeated matching, which is fundamentally different.
Our model is most aligned with that of \citet{caragiannis2024repeatedmatching}. However, they consider a slightly more general variant, whereby the value of an agent for a good in some round depends on the number
of rounds in which the good has been given to the agent in the past. They study approximately envy-free notions, show an intractability result, and special cases where fairness can be guaranteed. 
Our model, while more specialized than theirs, has a few distinctions: (i) we have stronger negative and intractability results, (ii) the fairness concept we consider is not envy-based, and is therefore novel in this domain, and (iii) we consider a notion of fairness \emph{at every round prefix}, something with prior work does not consider---they look at fairness at the end.
Recently, \citet{micheel2024repeatedhouse} also studied  essentially the same model (under a different name: \emph{repeated house allocation}), but with ordinal preferences and other kinds of envy-based measures.

\paragraph{Repeated fair division.}
\citet{igarashi2023repeatedfairallocation} studied a model of repeated fair division, where a set of goods is available at each round, and every good must be allocated. This is in contrast to our model where each agent gets exactly one good. They consider the compatibility of envy-freeness and Pareto optimality, and show positive results in restricted cases.
\citet{balan2011longtermfairness} study a similar model, but with a focus on the \emph{average} utility of goods received by the agents.
Note that as with classical fair division, house allocation (where each agent gets exactly one good) is a special case and has considerably different results. 
\citet{elkind2024temporalfairdivision} also consider a non-repeated (but also offline) variant of this model where a single good needs to be allocate at each round.

\paragraph{Multi-agent sequential decision-making.}
Several other works in multi-agent systems bear resemblance to our model.
For instance, \citet{zhang2014multiagentseq} also study the egalitarian objective multi-agent decision-making problems.
However, they take a non-cooperative game-theoretic approach and do not study a matching problem.
\citet{lim2024stochastic} consider an  assignment problem in the context of stochastic multi-armed bandits, with egalitarian fairness as the objective. In their setting, at each round, exactly one ``arm'' must be assigned to each user such that no two users are assigned to the same arm. However, the user's utility (``reward'') in this case is stochastic, and therefore explores a different problem.
Several other works \citep{cheng2005multiprocessorscheduling,kellerer1998onlinepartition} consider the problem of semi-online multiprocessor scheduling, with the objective of minimizing the \emph{makespan} (i.e., minimize the maximum time taken by any any processor). This is analogous to the egalitarian objective.
However, results in this setting only hold for identical valuations (since machines are identical), and primarily apply to a (semi-)online setting, where goods arrive one at a time (and so valuations over future goods are known not in advance), but the total valuation is known.

\paragraph{Santa Claus problem.}
Another related line of work is egalitarian fair division, also known as the \emph{Santa Claus} problem.
The standard model here is a single-shot fair division setting with an egalitarian objective, which was studied as far back as \citet{thomson1983egalitarian}, who axiomatically characterized the egalitarian solution using numerous desirable properties.
\citet{bansal2006santaclaus} then initiated the study of approximation algorithms for this problem, by providing an $\mathcal{O}(\log \log m/\log \log \log m)$ approximation algorithm for the special case when agents have \emph{restricted additive} valuations.
\citet{annamalai2015santaclaus} and \citet{davies2020santaclaus_constantapprox} subsequently provided a $12.33$- and $(4 + \varepsilon)$-approximation algorithm for this restricted case, respectively.
Numerous other works study \emph{online} variants of this problem, but typically under various relaxations---since strong worst-case guarantees are impossible without additional assumptions. Some of these restrictions include allowing for some reordering in the allocation process \citep{epstein2010reorderingbuffer} or restricting the number of agents
\citep{he2005twouniform,tan2006machinecovering_twouniform,wu2014semionlinemachinecovering}, or allowing transfer of items after assignment \citep{chen2011machinemigration}.

\paragraph{Other sequential decision-making models.}
We briefly mention several other models that may appear similar to (or could be superficially framed as) repeated matching, but are in fact distinct.
In the \emph{temporal voting} model \citep{alouf2022better,bulteau2021jrperpetual,chandak2024proportional,elkind2022temporalslot,elkind2024temporalelections,elkind2024temporalsurvey,elkind2025temporalchores,elkind2025verifying,phillips2025strengthening,zech2024multiwinnerchange}, the outcome is a sequence of decisions, where in each round a single project or candidate is selected.
These outcomes are \emph{public} in nature; they simultaneously benefit all agents rather than being individually allocated. 
While the same universal set of alternatives may exist across rounds (as in our model), the goal in temporal voting is to ensure fairness and representational balance across time in collective decisions. This differs fundamentally from repeated matching, where items are assigned \emph{exclusively} to individual agents in each round, and fairness arises from managing trade-offs in personal allocations over time.
Another related body of work looks at the \emph{online fair division} model \citep{aleksandrov2015onlinefoodbank,choo2025approximateproportionality,neoh2025onlineFDadditionalinfo,zhou2023icml_mms_chores}, where the repeated perspective does not apply, since a defining feature is uncertainty about future arrivals and valuations (the set of goods is not known in advance).

\section{Preliminaries} \label{sec:prelims}

Given a positive integer $z$, let $[z]=\{1,\dots,z\}$. We consider the problem of fairly matching a set of $n$ agents $N=[n]$ to a set of $m\geq n$ goods $G=\{g_{1},\dots,g_{m}\}$ over $T$ rounds.
We note that this is without loss of generality---to model the case of $m<n$, one can simply add zero-valued goods to arrive at the $m \geq n$ case and the results remain the same.

\paragraph{Matchings.} A \emph{matching} $M$ is an injective map from $N$ to $G$. We have $M(i)=g$ if and only if agent $i\in N$ is matched to good $g\in G$. In some instances, we also represent a matching either as a $n$-tuple $M=(M(1),\dots,M(n))$ or as an $n\times m$ matrix $M$, where $M_{ij}= 1$ if $M(i)=g_{j}$, and $0$ otherwise. We denote the set of all sequences of matchings with length at least $t\in[T]$ as $\bbS^{t}$. 

\paragraph{Valuations.} Let $u_{i}(g)$ denote the non-negative value that agent $i\in N$ receives when matched to good $g\in G$. The \emph{valuation profile} of a matching $M$ is the $n$-tuple $(u_{1}(M(1)),\dots,u_{n}(M(n)))$.
Given a sequence of $T$ matching $S=(M^{1},\dots,M^{T})$, the value that agent $i$ receives under $S$ up to round $t\in[T]$ is the sum of the values received up to that round, that is, $v_{i}^{t}(S)\coloneqq\sum_{s=1}^{t}u_{i}(M^{s}(i))$.

\paragraph{Instances.} An instance of the \emph{egalitarian repeated matching} problem is a tuple $\calI=\instfull$. The egalitarian (or maximin) objective seeks to maximize the value received by the worst-off agents. Let $t\in[T]$. We define the \emph{bottleneck agents} of a sequence $S\in\bbS^{t}$ at round $t$ as the set of agents who received the lowest value under $S$ up to that round. We further define the \emph{bottleneck value} as the value received by the bottleneck agents, that is, $b^{t}(S)\coloneqq\min_{i\in N}v_{i}^{t}(S)$.

\paragraph{Objective.} Motivated by the egalitarian objective, we denote the maximum bottleneck value at round $t$ as $\OPT(t)\coloneqq\max\{b^{t}(S)\,|\,S\in\bbS^{t}\}$. In this work, we consider two notions of optimality\footnote{For simplicity, we refer to \emph{optimality} as shorthand for the egalitarian welfare-maximizing optimal solution.}: one that ensures the best outcome at a specific round, and another that ensures the best outcome at every round up to a given round.
Both concepts of this nature (fairness at the end  or at the end of each prefix) have been studied in temporal/repeated fair division \citep{elkind2024temporalfairdivision,igarashi2023repeatedfairallocation} and repeated matching \citep{caragiannis2024repeatedmatching}.

We first introduce the weaker notion of optimality,\footnote{Note that our problem with optimality as an objective can be reformulated as a single-shot fair division problem with $T$ copies of each good and an added constraint that each agent receives exactly $T$ goods. While mathematically equivalent, this formulation is unintuitive in the classical setting, non-standard, and remains unexplored (with no known algorithms designed for it) in the literature. Furthermore, the \emph{sequential} perspective is necessary for defining and motivating anytime-optimality and enabling potential extensions, neither of which can be naturally accommodated in a single-shot optimization framework.}
which is defined by mandating fairness at the end of a particular round $t\in[T]$. More formally, we say that a sequence $S\in\bbS^{t}$ is \emph{optimal} at round $t\in[T]$ if $b^{t}(S)=\OPT(t)$.

Note that this property does not require optimality to hold at any previous rounds $s$, for $s<t$. However, for any round $t\in[T]$, if we require optimality at every round $s\leq t$, then we get a stronger notion of optimality. More formally, we say that a sequence $S\in\bbS^{t}$ is \emph{anytime optimal} up to round $t\in[T]$ if $b^{s}(S)=\OPT(s)$ for all rounds $s\in[t]$.

Observe that while anytime optimality is significantly stronger than standard optimality, positive results for anytime optimality do not necessarily extend to the well-studied online setting. This is because, in the online setting, goods typically arrive one at a time, and valuations over these goods can be arbitrary---potentially over an unlimited set.

\paragraph{Efficiency.} We also consider \emph{Pareto optimality}, a notion of economic efficiency commonly studied in the social choice literature. Formally, a matching $M$ is said to \emph{weakly Pareto dominates} another matching $M_{0}$ if all agents $i\in N$ receive at least as much value under $M$ as $M_{0}$, that is, $u_{i}(M(i))\geq u_{i}(M_{0}(i))$. A matching $M$ is said to \emph{strongly Pareto dominates} $M_{0}$ if $M$ weakly Pareto dominates $M_{0}$ and there exist some agent $i\in N$ with $u_{i}(M(i))>u_{i}(M_{0}(i))$. A matching $M$ is \emph{Pareto optimal} when no matching strongly Pareto dominates $M$.

\subsection{Allocations and Bistochastic Matrices}

Working with sequences of matchings can be challenging due to the constraints imposed by each matching. It would be helpful if we could ignore these constraints in our analysis and focus solely on the frequency with which each good is allocated to each agent. We refer to such an abstraction as an allocation. An \emph{allocation} $A=(A_{1},\dots,A_{n})$ is a collection of multiset, where $A_{i}$ is the multiset of goods that are allocated to agent $i\in N$. We can represent an allocation as a matrix $A$ where $A_{ij}$ is the number of times good $g_{j}\in G$ appears in $A_{i}$. The value that agent $i$ receives under $A$ is defined as
\begin{equation*}
    v_{i}(A)
    \coloneqq\sum_{\mathclap{g\in A_{i}}}u_{i}(g)
    =\sum_{\mathclap{g_{j}\in G}}A_{ij}u_{i}(g_{j}).
\end{equation*}

\Cref{lem:alloc-to-seq} states that an allocation can be transformed into a polynomial-length sequence of unique matching. Hence, when a proof is phrased in terms of allocations instead of a sequence, no generality is lost. Accordingly, we will often reason with allocations in our proofs, invoking the lemma whenever an explicit sequence of matchings is required.

\begin{restatable}{lemma}{alloctoseq}
    \label{lem:alloc-to-seq}
    Suppose $A\in\bbR^{n\times m}$ is an allocation with
    \begin{equation*}
        \sum_{\mathclap{i\in N}}A_{ij}\leq T\quad\text{and}\quad\sum_{\mathclap{g_{j}\in G}}A_{ij}\leq T.
    \end{equation*}
    Then, there exist a sequence of matchings $S$ consisting of $d\leq m^{2}-m+1$ unique matchings that satisfy $v_{i}^{T}(S)\geq v_{i}(A)$. This can be computed in polynomial time.
\end{restatable}

Several proofs of our results, including the preceding lemma, represent an allocation as a bistochastic matrix. A \emph{bistochastic matrix} is a non‑negative square matrix whose rows and columns each sum to 1, and a \emph{scaled integer bistochastic matrix} is its integer counterpart, with non‑negative integer entries and the sum of each row and column is a common integer. We defer an extended discussion of the mathematical preliminaries (along with all other omitted proofs in this paper) to the appendix.

\section{Finding Optimal Sequences} \label{sec:optimal-hardness}
We begin by focusing on optimality in this section.
We first show that finding an optimal sequence of matchings is computationally intractable. We then show an relationship between a multiplicative approximation to our problem and the popular Santa Claus problem.
Since computing exact solutions is intractable for large instances, we propose an approximation algorithm to find a near-optimal sequence efficiently. 
We also complement the hardness result by introducing a fixed-parameter tractable (FPT) algorithm that finds an optimal sequence when $n$ or $m$ is a constant, thereby providing an efficient algorithm for practical applications.

We assume that the reader is familiar with basic notions of classic complexity theory~\citep{papadimitriou_computational_2007} and parameterized complexity \citep{flum_parameterized_2006,niedermeier_invitation_2006}.

\subsection{Hardness Results}
\label{subsec:hardness}

Consider the decision problem associated with the egalitarian repeated matching problem, as follows. 
\begin{tcolorbox}[title=\textsc{Egalitarian Repeated Matching (ERM)}]
    \textbf{Input}: An instance $\instfull$ and a target $\kappa$.
    \tcblower
    \textbf{Question}: Is there a sequence $S\in\bbS^{T}$ with $b^{T}(S)\geq\kappa$? 
\end{tcolorbox}

We show that \textsc{ERM} is \NP-complete by reducing from a known \NP-hard problem, \textsc{3-occ-3-sat} (defined in the proof).
This result also implies that \textsc{ERM} is \APX-hard—that is, there exists no polynomial-time approximation scheme (PTAS) for the problem. Our result is as follows.
\begin{restatable}{theorem}{optimalhardness}
    \label{thm:optimal-hardness}
    \textsc{ERM} is \NP-complete (and \APX-hard) even when $u_i(g) \in \{0,0.5,1\}$ for all $i \in N$ and $g \in G$, for any $T\geq 2$.
\end{restatable}

An implication of \textsc{ERM} not having a PTAS is that only constant-factor multiplicative approximations may be possible (though its existence is not guaranteed). We define this formally: for any $c \in [1,\infty)$, we say that an algorithm is $c$-\emph{approximate} (or simply \emph{$c$-approx}) if the sequence $S\in\bbS^{t}$ returned by the algorithm satisfy $b^{t}(S)\geq\OPT/c$ for all $t\in[T]$. When $c=1$, we have an exact algorithm.
A natural question is whether \textsc{ERM} admits a $c$-approx algorithm, for some constant $c \in [1,\infty)$.
Interestingly, we show that the existence of a $c$-approx algorithm for \textsc{ERM} would imply the existence of a $c$-approx algorithm for the single-shot egalitarian fair division problem (i.e., the Santa Claus problem with additive valuations\footnote{We specify ``additive valuations'' explicitly as some works (e.g., \citet{davies2020santaclaus_constantapprox}) consider a more restricted variant of the Santa Claus problem with restricted additive valuations.}).

\begin{restatable}{proposition}{SCred}
    \label{prop:SC-red}
    For any $c \in [1,\infty)$, there is a $c$-approx algorithm for ERM only if there is a $c$-approx algorithm for the Santa Claus problem with additive valuations.
\end{restatable}

The result above implies that finding even a constant-factor multiplicative approximation algorithm for \textsc{ERM} is likely to be very challenging.
This is because, despite the Santa Claus problem being a well-studied and long-standing problem, no constant-factor approximation is currently known for the version with general additive valuations. A constant-factor approximation is only known in the restricted additive case.\footnote{The current best known approximation factor is $(4 + \varepsilon)$, for a small $\varepsilon > 0$ in this restricted case \citep{davies2020santaclaus_constantapprox}.}

\subsection{Approximation Algorithm}
\label{subsec:approxalgo}

Given the results above, we focus on whether we can achieve an \emph{additive} approximation with respect to optimality instead.
We now describe an approximation algorithm that achieves an additive approximation bound independent of the number of rounds $T$. Crucially, this implies that as $T$ increases, the approximate solution converges rapidly to the optimal one.
The setting when the number of rounds is large can be observed in applications where the matching process runs continuously over extended periods---such as dynamic spectrum allocation (where the system operates continuously, often measured in (milli)seconds), leading to an immense number of allocation rounds.

Without loss of generality, we can assume that $n=m$; otherwise, we can simply create $m-n$ dummy agents with $u_{i}(g_{j})=\max_{i'\in N}\max_{g'_{j}\in G}u_{i'}(g'_{j})$ for all dummy agents $i$ and goods $g$. Then, consider the following linear program:

\begin{subequations}
    \begin{alignat}{2}
    \maximize_{b,B} \quad & b\tag{P1}\label{eqn:approx-general-lp}\\
    \subjectto \quad & \sum_{\mathclap{g_{j}\in G}} B_{ij}u_{i}(g_{j})\geq b, && \quad\forall i\in N,\nonumber\\
    & \sum_{\mathclap{g_{j}\in G}} B_{ij}=1, && \quad\forall i\in N,\nonumber\\
    & \sum_{i\in N} B_{ij}=1, && \quad\forall g_{j}\in G,\nonumber\\
    & B_{ij} \geq 0, && \quad\forall i\in N,\,\forall g_{j}\in G.\nonumber
    \end{alignat}
\end{subequations}
Note that the solution to (\ref{eqn:approx-general-lp}) is a bistochastic matrix $B$. Our approximation algorithm uses Birkhoff's algorithm to decompose $B$ into a convex combination of matchings. The number of times each matchings are included in the sequence is then determined by the convex coefficients (see \Cref{alg:approx-general}).

\begin{algorithm}[ht]
    \caption{Approximation algorithm for finding an optimal sequence of matchings}
    \label{alg:approx-general}
    \textbf{Input}: An instance $\calI=\instfull$
    
    \begin{algorithmic}[1]
        \STATE let $B$ be the solution to linear program (\ref{eqn:approx-general-lp})
        \STATE decompose $B$ into $\alpha_{1}M_{1}+\dots+\alpha_{d}M_{d}$ using \Cref{alg:birkhoff}
        \STATE let $S$ be an empty sequence
        \STATE add $\lfloor T\alpha_{k}\rfloor$ copies of $M_{k}$ in $S$ for each $k\in[d]$
        \STATE add any matchings into $S$ so that $|S|=T$
        \STATE \textbf{return} $S$ 
    \end{algorithmic}
\end{algorithm}

Then, we prove the following result.
\begin{restatable}{theorem}{approxgeneral}
    \label{thm:approx-general}
    Given an instance $\instfull$, the sequence $S\in\bbS^{T}$ returned by \Cref{alg:approx-general} satisfy
    \begin{equation*}
        b^{T}(S)\geq\OPT(T)-m\cdot\max_{i\in N}\max_{g\in G}u_{i}(g).
    \end{equation*}
\end{restatable}
\begin{proof}
    Consider the allocation $A$ in which $A_{ij}=\lfloor TB_{ij}\rfloor$ for all $i\in N$ and $g_{j}\in G$. Note that for each $g_{j}\in G$, we have
    \begin{equation*}
        \sum_{\mathclap{i\in N}}A_{ij}
        =\sum_{\mathclap{i\in N}}\,\lfloor TB_{ij}\rfloor
        \leq\sum_{\mathclap{i\in N}}TB_{ij}
        =T,
    \end{equation*}
    and similarly, for each $i\in N$, we have
    \begin{equation*}
        \sum_{\mathclap{g_{j}\in G}}A_{ij}
        =\sum_{\mathclap{g_{j}\in G}}\,\lfloor TB_{ij}\rfloor
        \leq\sum_{\mathclap{g_{j}\in G}}TB_{ij}
        =T.
    \end{equation*}
    By \Cref{lem:alloc-to-seq}, there exist a sequence $S$ over $T$ rounds composed of at most $O(m^{2})$ unique matchings such that $v_{i}^{T}(S)\geq v_{i}(A)$. Then, for any agent $i\in N$, we have
    \begin{align*}
        v_{i}^{T}(S)
        \geq v_{i}(A)
        &\geq\sum_{\mathclap{g_{j}\in G}}u_{i}(g_{j})\lfloor TB_{ij}\rfloor\\
        &\geq\sum_{\mathclap{g_{j}\in G}}u_{i}(g_{j})\cdot(TB_{ij}-1)\\
        &=\sum_{\mathclap{g_{j}\in G}}TB_{ij}u_{i}(g_{j})-\sum_{\mathclap{g_{j}\in G}}u_{i}(g_{j})\\
        &\geq T b-m\cdot\max_{\mathclap{g_{j}\in G}}u_{i}(g_{j})\\
        &\geq\OPT(T)-m\cdot\max_{\mathclap{g_{j}\in G}}u_{i}(g_{j}).
    \end{align*}
    Let $k\in N$ be a bottleneck agent of sequence $S$ at round $T$ so that $b^{T}(S)=v_{k}^{T}(S)$. Then, we have
    \begin{equation*}
        b^{T}(S)
        \geq\OPT(T)-m\cdot\max_{\mathclap{g_{j}\in G}}u_{k}(g_{j}) \geq\OPT(T)-m\cdot\max_{\mathclap{i\in N}}\max_{\mathclap{g_{j}\in G}}u_{i}(g_{j}). \qedhere
    \end{equation*}

\end{proof}
Note that although the maximum valuation can be arbitrarily large, they are typically bounded in practice. Consequently, such a bound remains informative and relevant. Instance‑dependent additive bounds of this type are well-established in the literature, particularly in the context of stochastic bandits \citep{lattimore2020bandits,lim2024stochastic} and online fair division \citep{benade2018envyvanish,hajiaghayi2022santaclaus}.

\subsection{Fixed-Parameter Tractable (FPT) Algorithm}
\label{subsec:fptalgo}

Next, we consider another approach to dealing with computational intractability.
We show that the problem is \emph{fixed parameter tractable} (FPT) when the number of agents is a fixed parameter, i.e., there exists an algorithm that can compute an optimal sequence in polynomial-time when $n$ is a constant. This provides a practical solution for small-group matching. Our result is as follows.
\begin{restatable}{theorem}{fptgeneral}
    \label{thm:fpt-general}
    Given an instance $\instfull$, \textsc{ERM} is FPT with respect to $n$.
\end{restatable}

The proof of \Cref{thm:fpt-general} relies on our newly established characterizations of Pareto‑optimal and maximum matchings in terms of permutations of agents. These results may be of independent interest to researchers in matching and house allocation.

In particular, let $\pi:N\to[n]$ be a permutation of the agents. A matching $M_{*}$ is said to be \emph{$\pi$-optimal} if there exists no matching $M$ such that
\begin{itemize}
    \item Some agent $i\in N$ satisfies $u_{i}(M(i))>u_{i}(M_{*}(i))$; and
    \item For every such agent $i$, it holds that for all agents $i'\in N$ with $\pi(i')<\pi(i)$, we have $u_{i'}(M(i'))\geq u_{i'}(M_{*}(i'))$.
\end{itemize}
Then, we obtain the following lemma.
\begin{restatable}{lemma}{paretocharacterization}
    \label{lem:pareto-characterization}
    A matching $M$ is Pareto optimal if and only if it is $\pi$-optimal for some permutation $\pi$.
\end{restatable}

In the context of house allocation without indifferences, it is well-established that \emph{serial dictatorship} characterizes Pareto-optimal allocations \citep{abdulkadirouglu1998random}. However, when agents are allowed to express indifferences between houses, the allocations produced by serial dictatorship are not guaranteed to be Pareto optimal \citep{abraham2004pareto}. Therefore, our definition of $\pi$-optimal can be interpreted as an extension of serial dictatorship that ensures Pareto optimality even in the presence of indifferences.

We describe how this characterization leads to an FPT algorithm in \Cref{sec:fptalgo}.

\section{Anytime Optimality} \label{sec:anytime-optimality}
In this section, we consider the problem of \emph{anytime optimality}, a stronger notion that requires optimality at every round prefix. We show that an anytime optimal sequence always exists when $n=2$, but determining whether such a sequence exists for $n\geq 3$ is \coNP-hard.
The setting of $n=2$ is a widely studied and is an important special case in related literature \citep{elkind2024temporalfairdivision,gollapudi2020matching,igarashi2023repeatedfairallocation}.
Our results are as follows.

\begin{restatable}{theorem}{anytimetwoagents}
    \label{thm:anytime-twoagents}
    Given an instance $\instfull$ with $n=2$, there always exist an anytime optimal sequence of matchings, and we can find it in polynomial time.
\end{restatable}

However, we show that this positive result does not extend to the case when $n \geq 3$, for all $T \geq 2$, with the following impossibility result.
\begin{restatable}{proposition}{anytimedontexist}
    \label{prop:anytime-exist}
    An anytime optimal sequence might not exist for any problem instance $\instfull$ with $n\geq 3$ and $T\geq 2$.
\end{restatable}

\begin{proof}
    Consider the following instance with $m=n=3$. For each $i\in N$ and $g_{j}\in G$, let $u_{i}(g_{j})=U_{ij}$, where
    \begin{equation*}
        U=
        \begin{bmatrix}
            5 & 2 & 1 \\
            3 & 3 & 2 \\
            2 & 5 & 1
        \end{bmatrix}.
    \end{equation*}
    Note that $\OPT(1)=2$ and $\OPT(2)=6$. Furthermore, the only way to achieve $\OPT(2)$ is by choosing $M_{1}=(1,2,3)$ and $M_{2}=(3,1,2)$ in any order. As such, the bottleneck value at $t=1$ is $1$, which is not anytime optimal.
\end{proof}

The above implies that we cannot hope for anytime optimality in most cases.
However, given a problem instance, one may still wish to obtain an anytime optimal result \emph{if it exists}.
Unfortunately, we show that even determining whether an instance admits an anytime optimal solution is computationally intractable, with the following result.

\begin{restatable}{theorem}{anytimehardness}
    \label{thm:anytime-hardness}
    Given instance $\calI=\instfull$, the problem of deciding if $\calI$ admits an anytime optimal sequence is \coNP-hard.
\end{restatable}

Finally, we complement the above hardness result with an approximation algorithm that achieves an additive approximation bound independent of the number rounds $T$. Again, this means that as $T$ increases, the approximate solution converges rapidly to the optimal one.

\begin{algorithm}[ht]
    \caption{Approximate algorithm for anytime optimal sequence}
    \textbf{Input:} An instance $\calI=\instfull$
    
    \label{alg:approx-anytime}
    \begin{algorithmic}[1]
        \STATE let $B$ be the solution to (\ref{eqn:approx-general-lp})
        \STATE decompose $B$ into $\alpha_{1}M_{1}+\dots+\alpha_{d}M_{d}$ using \Cref{alg:birkhoff}
        \STATE initialize $n_{k}=0$ for all $k\in[d]$
        \FOR{$t=1,\dots,T$}
        \STATE choose matching $M^{t}=\arg\min_{M_{k}}(n_{k}+1)/\alpha_{k}$
        \STATE update $n_{k}\leftarrow n_{k}+1$
        \ENDFOR
        \STATE \textbf{return} $\{M_{1},\dots,M_{T}\}$
    \end{algorithmic}
\end{algorithm}
\begin{restatable}{theorem}{anytimeapprox}
    \label{thm:anytime-approx}
    Given an instance $\instfull$, there always exist a sequence of matchings that is approximate anytime optimal. Furthermore, \Cref{alg:approx-anytime} outputs a sequence of matchings $S$, in polynomial time, that satisfy
    \begin{equation*}
        b^{t}(S)\geq\OPT(t)-5m\cdot\max_{i\in N}\max_{g\in G}u(g), \quad\forall t\in[T].
    \end{equation*}
\end{restatable}

\begin{proof}[Proof sketch]
    Let $n_{kt}$ be the value of $n_{k}$ after round $t$. After each round $t\in[T]$, we claim that our choice of matching $M^{t}$ maintains the invariant $n_{kt}\geq\alpha_{k} \cdot t-1$ for all $k\in[d]$. Intuitively, this says that by any round $t$, each matching $M_{k}$ has been selected for roughly its intended $\alpha_{k}$ fraction of the rounds. Thus, we will get a result similar to that of \Cref{thm:approx-general}. More specifically, we can show that $v_{i}^{t}(S)
        \geq\OPT(t)-d\cdot\max_{g\in G}u_{i}(g)$
    for all $i\in N$. Observe that since (\ref{eqn:approx-general-lp}) has $m^{2}+5m$ inequality constraints and $m^{2}+1$ variables, $m^{2}+1$ constraints will be tight at a vertex solution, meaning there are at most $5m$ non-zero entries in $B$, which implies that $d\leq 5m$. 
\end{proof}

\section{Special Cases} \label{sec:special-cases}

In this section, we shift our focus back to optimality\footnote{Unfortunately, anytime optimality is a strong condition with relatively strong negative results (as with many similar problems in the online setting). We leave the existence (or impossibility) of obtaining anytime optimality in special cases as an interesting direction for future work.} and consider three special cases: (1) when agents have binary valuations, (2) when there are only two types of goods, and (3) when agents share identical valuations.
For each of the first two cases, we provide an algorithm that computes an optimal sequence of matchings in polynomial time.
A key technique that we used here is to reduce the problem to one of \emph{circulation with demand} and leveraging the Ford-Fulkerson algorithm to compute a feasible circulation. We then show for the third case that the problem is hard even for optimality, address a special case where we it can be solved in polynomial time, and provide an approximate anytime optimal algorithm for it.

\subsection{Binary Valuations}
The first setting we consider is when agents have \emph{binary valuations}, i.e. $u_{i}:G\to\{0,1\}$ for all agents $i\in N$.
This is an important and well-studied subclass of valuations (sometimes referred to as \emph{binary additive} valuations).
Numerous fair division \citep{aleksandrov2015onlinefoodbank,amanatidis2021mnwefx,bouveret2016conflict,freeman2019equitable,HalpernPrPs20,hosseini2020infowithholding,SuksompongTe22} and matching \citep{BogomolnaiaMo04,gollapudi2020matching} papers consider this setting.
Binary valuations can also be viewed as \emph{approval votes}, which have long been
studied in the voting literature \citep{brams2007approvalvoting,kilgour2010approval}, and permit very simple elicitation.

Notably, under binary valuations, maximizing egalitarian welfare is equivalent to maximizing \emph{Nash welfare} (i.e., the geometric mean), which is an extremely popular concept in fair division, and has many desirable properties \citep{HalpernPrPs20,SuksompongTe22}.

We first establish the following lemma. 
\begin{restatable}{lemma}{matchingpolem} \label{lem:matching_po}
    Let $G'$ be the goods in a maximum matching. Then, for any matching $M$, there is a matching $M_{*}$ that weakly Pareto dominates $M$ and that the goods matched by $M_{*}$ is a subset of $G'$.
\end{restatable}

The above lemma basically provides another characterization, this time, of maximum matchings under binary valuations in terms of Pareto optimality. To the best of our knowledge, this result is also novel in the context of house allocation, which may be of independent interest. This lemma is used to prove the following result.

\begin{restatable}{theorem}{binaryoptimal}
    \label{thm:binary-optimal}
    Given an instance $\instfull$ with binary valuations, we can find an optimal sequence of matchings in polynomial time.
\end{restatable}

Note that the NP-hardness result of \Cref{thm:optimal-hardness} implies that we cannot strengthen the positive results to the setting where agents have \emph{ternary valuations} (or three-valued instances) \citep{fitzsimmons2025ternary}.

\subsection{Two Types of Goods}

Next, we consider the setting with two \emph{types} of goods: each good can be divided into two groups, and each agent values all goods in a particular group equally. 
This preference restriction is also commonly studied in (temporal) fair division \citep{ALRS23,elkind2024temporalfairdivision,GMQ24}. Formally, let $G_{0},G_{1}\subseteq G$ be a partition of the set of goods such that $G_{0}\cap G_{1}=\emptyset$, $G_{0}\cup G_{1}=G$, and for all agent $i\in N$ and all goods $g,g'\in G_{r}$ for some $r\in\{0,1\}$, we have $u_{i}(g)=u_{i}(g')$.
Then, our result is as follows.

\begin{restatable}{theorem}{optimaltwogoods}
    Given an instance $(N,G_{1}\cup G_{2},T,\{u_{i}\}_{i\in N})$ with two types of goods, we can find an optimal sequence of matchings in polynomial time.
\end{restatable}

\subsection{Identical Valuations}
The last special case we consider here is one where agents have \emph{identical valuation} functions, i.e., $u_{i}=u_{i'}$ for all agents $i,i'\in N$.
The setting with identical valuations is also well-studied in the repeated fair division/matching \citep{caragiannis2024repeatedmatching,igarashi2023repeatedfairallocation} and standard fair division \citep{barman2020identical,mutzari2023resilient,plaut2020almost} literature. 
Moreover, works on semi-online multiprocessor scheduling with the makespan minimization objective (analogous to the egalitarian objective) \citep{cheng2005multiprocessorscheduling,kellerer1998onlinepartition} focus on identical valuations as well (since machines are identical in that setting). 

We show that even under this restricted setting of identical valuations, the problem of finding an optimal sequence is generally still \NP-hard.

\begin{restatable}{theorem}{optimalidenticalhardness}
    \label{thm:optimal-identical-hardness}
    Given an instance $\instfull$ with identical valuations, finding an optimal sequence of matchings is \NP-complete.
\end{restatable}

However, when $T$ is a multiple of $n$, we shown that the problem can be solved in polynomial time, with the following two results. We note that the case when $T$ is a multiple of $n$ is also a popular special case studied in repeated matching/fair division \citep{caragiannis2024repeatedmatching,igarashi2023repeatedfairallocation}

\begin{restatable}{theorem}{optimalidenticalfactor}
    \label{thm:optimal-identical-factor}
    Given an instance $\instfull$ with identical valuations and $T=kn$ for some $k\in\bbZ$, we can find an optimal sequence of matchings in polynomial time.
\end{restatable}
Finally, we complement the above with an approximation algorithm that achieves (even anytime) optimality up to an additive approximation factor of $\max_{g\in G} u(g)$.\footnote{We denote agents' identical utility function as $u$. Then, $v_{i}^{t}(S)\coloneqq\sum_{s=1}^{t} u(M^{s}(i))$ for all $t\in[T]$.}
This gives us a stronger result compared to the general case, which is also only for optimality (as in \Cref{thm:approx-general}).

\begin{restatable}{theorem}{approxidentical}
    \label{thm:approx-identical}
    Given an instance $\instfull$ with identical valuations, we can find, in polynomial time, a sequence of matchings $S$ that satisfy
    \begin{equation*}
        b^{t}(S)\geq\OPT(t)-\Delta,\quad \forall t\in[T],
    \end{equation*}
    where $\Delta$ is the difference in value between the most valuable good and the $n$-th most valuable good.\footnote{This is equivalent to the concept of \emph{gap} in the bandits literature.} 
\end{restatable}

\section{Conclusion}
In this work, we introduced and studied a model of repeated matching with goal of obtaining egalitarian optimality.
We investigated the computational complexity of achieving optimality and anytime optimality, and identified several settings where these problems can be solved efficiently, together with accompanying algorithms.
Specifically, for optimality, we provided an approximation algorithm independent of $T$, and FPT algorithms with respect to $n$ or $m$. For anytime optimality, we provided an approximation algorithm that complements the hardness and impossibility result even in simple cases.
We also showed two special cases (binary valuations, two types of goods) where optimality can be achieved, and a final special case (identical valuations) where approximate anytime optimality can be achieved. 

Directions for future work include considering other special cases that admit efficient optimal solutions, such as bi-valued utilities (where each agent values each good at either $1$ or some integer $p > 1$) or identical rankings.
It would also be interesting to study concepts that interpolate optimality and anytime optimality (e.g., optimality at every $\tau$ timesteps).
In two of our special cases, we mentioned the equivalence between egalitarian and Nash welfare. It would be interesting to identify the conditions under which these two objectives are equivalent in this setting.

\bibliographystyle{plainnat}
\bibliography{aaai2026}

\appendix
\include{arxiv-appendix}

\end{document}

%% file: macros.tex
\def\instfull{\smash{(N,G,T,\{u_{i}\}_{i\in N})}}

\def\calI{\mathcal{I}}

\def\bbN{\mathbb{N}}
\def\bbR{\mathbb{R}}
\def\bbS{\mathbb{S}}
\def\bbZ{\mathbb{Z}}

\def\NP{\textsf{NP}}
\def\coNP{\textsf{coNP}}
\def\APX{\textsf{APX}}

\def\OPT{\textsf{OPT}}

\DeclareMathOperator*{\maximize}{maximize}
\DeclareMathOperator*{\subjectto}{subject\ to}

\tcbset{size=small,top=1mm,bottom=1mm,middle=1mm}

\newtheorem{theorem}{Theorem}[section]

\newtheorem{lemma}[theorem]{Lemma}

\theoremstyle{definition}
\newtheorem{definition}[theorem]{Definition}

\theoremstyle{remark}

%% file: arxiv-appendix.tex
\onecolumn
\begin{center}
    \LARGE\bf
    Appendix
\end{center}
\vspace{2mm}

\section{Allocations and Bistochastic Matrices}
\label{sec:proofs}

We devote this section to establishing \Cref{lem:alloc-to-seq}.


We begin by showing that any allocation $A$ satisfying the conditions of \Cref{lem:alloc-to-seq} can be transformed into a scaled integer bistochastic matrix via \Cref{alg:alloc-to-bistoc}.

\begin{algorithm}[h]
    \caption{Convert allocation to a scaled integer BM}
    \textbf{Input:} Allocation $A$ and integer $T$
    \label{alg:alloc-to-bistoc}
    \begin{algorithmic}[1]
        \STATE let $B$ be a copy of $A$
        \STATE append $B$ with $m-n$ rows of zeros
        \WHILE{some row $i$ and column $j$ does not sum to $T$}
            \STATE increment $B_{ij}'$ by
            \vspace{-0.25cm}
            \begin{equation*}
                T-T\cdot\max\Biggl\{\,\sum_{i'=1}^{m}B_{i'j}\,,\,\sum_{j'=1}^{m}B_{ij'}\,\Biggr\}
            \end{equation*}
            \vspace{-0.25cm}
        \ENDWHILE
        \STATE \textbf{return} $B$
    \end{algorithmic}
\end{algorithm}

\begin{lemma}
    \label{lem:alloc-to-bistoc}
    Suppose $A\in\bbR^{n\times m}$ is an allocation with
    \begin{equation*}
    \sum_{\mathclap{i\in N}}A_{ij}\leq T\quad\text{and}\quad\sum_{\mathclap{g_{j}\in G}}A_{ij}\leq T.
    \end{equation*}
    There exist a scaled integer bistochastic matrix $B\in\bbR^{m\times m}$ such that the sum of each row and column is $T$, and for each agent $i\in N$, $v_{i}(B)\geq v_{i}(A)$. This can be computed by \Cref{alg:alloc-to-bistoc} in polynomial time.
\end{lemma}

\begin{proof}
    Observe that \Cref{alg:alloc-to-bistoc} must exit its loop after at most $2m$ iterations. This is because after each iteration, at least one more row or column will sum to $T$, and no row or column will sum to greater than $T$. Then, after at most $2m$ iterations, every row and column must sum to $T$. Since each iteration can be completed in polynomial time, the loop will also terminate in polynomial time. Furthermore, by the loop condition, we know that once the algorithm exits the loop, $B$ will be a scaled integer bistochastic matrix.
\end{proof}

Given a bistochastic matrix $B$, the \emph{Birkhoff-von Neumann Theorem}~\citep{birkhoff1946} states that $B$ can be written as a convex combination of $d\leq m^{2}-m+1$ matchings: $B=\alpha_{1}M_{1}+\dots+\alpha_{d}M_{d}$,
where $\alpha_{1},\dots,\alpha_{d}$ are non-negative coefficients that sums to $1$, and $M_{1},\dots,M_{d}$ are matchings. This decomposition can be computed in polynomial time using \emph{Birkhoff's algorithm}, and we describe the procedure for computing the coefficients and matchings in \Cref{alg:birkhoff}. \Cref{lem:scaled-bistoc-birkhoff} further extend the Birkhoff-von Neumann Theorem to scaled integer bistochastic matrices.

\begin{algorithm}[ht]
    \caption{Birkhoff's algorithm}
    \label{alg:birkhoff}
    \textbf{Input:} (Scaled integer) bistochastic matrix $B\in\bbR^{m\times m}$
    
    \begin{algorithmic}[1]
        \STATE let $B'$ be a copy of $B$
        \STATE initialize $k=1$
        \WHILE{there are non-zero entries in $B'$}
        \STATE construct a bipartite graph $G=([m],[m],E)$ such that $(i,j)\in E$ if and only if $B'_{ij}>0$
        \STATE find a perfect matching $E'\subseteq E$ of $G$
        \STATE let $\alpha_{k}=\min\{B'_{ij}\mid(i,j)\in E'\}$
        \STATE let $M_{k}$ be the permutation matrix corresponding to the perfect matching $E'$
        \STATE update $B'\leftarrow B'-\alpha_{k}M_{k}$
        \STATE update $k\leftarrow k+1$
        \ENDWHILE
        \STATE \textbf{return} $\{\alpha_{1},\dots,\alpha_{d}\}$ and $\{M_{1},\dots,M_{d}\}$
    \end{algorithmic}
\end{algorithm}

\begin{lemma}
    \label{lem:scaled-bistoc-birkhoff}
    If $B\in\bbR^{m\times m}$ is a scaled integer bistochastic matrix, then we can decompose $B$ as a linear combination of $d\leq m^{2}-m+1$ matchings $M_{1},\dots,M_{d}$ with positive integer coefficients $\alpha_{1},\dots,\alpha_{d}$: $B=\alpha_{1}M_{1}+\dots+\alpha_{d}M_{d}$.
    This can be computed by \Cref{alg:birkhoff} in polynomial time.
\end{lemma}

\begin{proof}
    This decomposition can be achieved by Birkhoff's algorithm, and its correctness follows analogously to the proof of correctness for Birkhoff's algorithm. As such, we will focus only on showing that $\alpha_{1}\dots,\alpha_{d}$ are positive integers. We claim that at the start of each iteration of the loop, the entries in $A'$ can only be non-negative integers. This is trivially true in the first iteration. Suppose this is true for the $k$-th iteration. Since the coefficient $\alpha_{k}$ is the minimum entry of $A'$ that corresponds to the perfect matching $E'$, $\alpha_{k}$ is a positive integer. Furthermore, after we update $A'$ by subtracting $\alpha_{k}$ from the entries that correspond to the matching $E'$, they must remain as non-negative integers. Thus, at the start of iteration $k+1$, the entries in $A'$ can only be non-negative integers. Since the coefficient $\alpha_{k}$ is just some entry of $A'$ at iteration $k\in[d]$, it is a positive integer.
\end{proof}

\Cref{lem:alloc-to-bistoc} and \Cref{lem:scaled-bistoc-birkhoff} together imply that an allocation $A$ can be transformed into a sequence of matchings $S$ by first converting $A$ into a bistochastic matrix using \Cref{alg:alloc-to-bistoc}, then next applying \Cref{alg:birkhoff} to convert the bistochastic matrix into a sequence. The resulting sequence satisfies the inequality $v_{i}^{T}(S)\geq v_{i}(A)$.

\begin{proof}[Proof of \Cref{lem:alloc-to-seq}]
    The claim follows directly from \Cref{lem:alloc-to-bistoc} and \Cref{lem:scaled-bistoc-birkhoff}.
\end{proof}

\section{Hardness Results for Optimality (\Cref{subsec:hardness})}
\label{proofs_secoptimal-hardness}

\optimalhardness*

\begin{proof}
    We will utilize the following decision problem \textsc{3-occ-3-sat} that is known to be \NP-hard.
    \begin{tcolorbox}[title=\textsc{3-occurrences 3-satisfiability (3-occ-3-sat)}]
    \textbf{Input}: A boolean formula $\Phi$ with $p$ variables $x_{1},\dots,x_{p}$ and $q$ clauses $c_{1},\dots,c_{q}$. For each $i\in[p]$, the literal $x_{i}$ appears twice and the literal $\bar{x}_{i}$ appears once.
    \tcblower
    \textbf{Question}: Is there an assignment for the variables such that $\Phi$ evaluates to \texttt{TRUE}? 
\end{tcolorbox}
    
    We first prove that \textsc{ERM} is \NP-hard when $T = 2$ and $u_i(g) \in \{0,0.5,1\}$ for all $i \in N$ and $g \in G$.
    Given a \textsc{3-occ-3-sat} instance $\Phi$, we will reduce it to an \textsc{ERM} instance $(\calI,\kappa)$ with $\calI=\instfull$, $T=2$, and $\kappa=1$. We construct $\calI$ as follows: For each $i\in[p]$, create three agents who are labeled, by an abuse of notation, as $x_{i1},x_{i2},x_{i3}$ and three goods $g_{i1},g_{i2},g_{i3}$. The base valuation of these agents are given by
    \begin{equation*}
        \begin{matrix}
            & g_{i1} & g_{i2} & g_{i3} \\[4pt]
            x_{i1} & 0.5 & 0.5 & 0.0 \\
            x_{i2} & 0.0 & 0.0 & 1.0 \\
            x_{i3} & 0.5 & 0.5 & 1.0
        \end{matrix}
    \end{equation*}
    and $0$ otherwise. Then, for each $j\in[q]$, create one agent with label $c_{j}$. The base valuation of $c_{j}$ is given by $u_{c_{j}}(g)= 1$ if literal $x_{i}$ is in clause $c_{j}$, and $g=g_{i1}$ or $g=g_{i2}$; or if literal $\bar{x}_{i}$ is in clause $c_{j}$ and $g=g_{i3}$; and $u_{c_{j}}(g)= 0$ otherwise.
    We create another $3p+q$ dummy goods, each with zero value for all agents.

    Suppose $\Phi$ is a \texttt{YES} instance and let $x_{1},\dots,x_{p}$ be a satisfying assignment. Consider the sequence $S$ constructed as follows: For each $i\in[p]$, if $x_{i}$ is \texttt{TRUE}, we match agents $x_{i1},x_{i2},x_{i3}$ according to
    \begin{equation*}
        \begin{matrix}
            & & t=1 & & \quad & & t=2 & \\
            & g_{i1} & g_{i2} & g_{i3} & \quad & g_{i1} & g_{i2} & g_{i3} \\[4pt]
            x_{i1} & 0 & 1 & 0 & \quad & 1 & 0 & 0\\
            x_{i2} & 0 & 0 & 1 & \quad & 0 & 0 & 0 \\
            x_{i3} & 0 & 0 & 0 & \quad & 0 & 0 & 1
        \end{matrix}
    \end{equation*}
    where we match an agent to a dummy good if its row contains all zeros. Since there are exactly two clauses, say $c_{a}$ and $c_{b}$, with literal $x_{i}$, we can, if necessary, match agent $c_{a}$ to good $g_{i1}$ at $t=1$ and match agent $c_{b}$ to good $g_{i2}$ at $t=2$. Likewise, if $x_{i}$ is \texttt{FALSE}, then we match according to
    \begin{equation*}
        \begin{matrix}
            & & t=1 & & \quad & & t=2 & \\
            & g_{i1} & g_{i2} & g_{i3} & \quad & g_{i1} & g_{i2} & g_{i3} \\[4pt]
            x_{i1} & 0 & 1 & 0 & \quad & 1 & 0 & 0\\
            x_{i2} & 0 & 0 & 1 & \quad & 0 & 0 & 0 \\
            x_{i3} & 1 & 0 & 0 & \quad & 0 & 1 & 0
        \end{matrix}
    \end{equation*}
    Since there are exactly one clause, say $c_{a}$, with literal $\bar{x}_{i}$, we can, if necessary, match agent $c_{a}$ to good $g_{i3}$ at $t=2$. It is easy to verify that each agent $i\in N$ has $v_{i}^{2}(S)\geq 1$. Thus, $(\calI,\kappa)$ is a \texttt{YES} instance.

    Suppose $(\calI,\kappa)$ is a \texttt{YES} instance and let $S$ be a solution to this instance. For each $j\in[q]$, if the literal $x_{i}$ appears in clause $c_{j}$ and agent $c_{j}$ is matched to either good $g_{i1}$ or $g_{i2}$, then we set $x_{i}$ to \texttt{TRUE} to satisfy the clause. Similarly, if the literal $\bar{x}_{i}$ appears in clause $c_{j}$ and agent $c_{j}$ is matched to good $g_{i3}$, then we set $x_{i}$ to \texttt{FALSE} to satisfy the clause. Since $v_{c_{j}}^{2}(S)\geq 1$, at least one literal in clause $c_{j}$ must have been set to satisfy the clause.
    
    This procedure might be ambiguous because two clauses might assign different values to the same variable. We claim that this will not happen. Suppose, for sake of contradiction, that there exist two clauses $c_{a}\neq c_{b}$ such that the literal $x_{i}$ appears in clause $c_{a}$ while the literal $\bar{x}_{i}$ appears in $c_{b}$, and that agent $c_{a}$ is matched to, without loss of generality, good $g_{i1}$ while agent $c_{b}$ is matched to good $g_{i3}$. Then, agent $x_{i2}$ must be matched to good $g_{i3}$ once, leaving only one copy of good $g_{i1}$ and two copies of good $g_{i2}$ to be allocated to agent $x_{i1}$ and $x_{i3}$. However, there is no way to achieve $b^{2}(S)\geq 1$ with this configuration, which leads to a contradiction.

    We now prove the case for $T\geq 3$. We will perform a reduction from an instance of the decision problem $(\calI',\kappa')$ where $\calI'=(N',G',T',\{u'_{i}\}_{i\in N'})$ with $T'=2$.

    Given an instance $(\calI',\kappa')$, we will reduce it to an instance $(\calI,\kappa)$, where $\kappa=\kappa'+C$ with
    \begin{equation*}
        C=T\cdot\max_{i\in N'}\max_{g_{j}\in G'}u'_{i}(g_{j})
    \end{equation*}
    and $\calI=\instfull$ with $n=m=2n'+m'$ and
    \begin{equation*}
        u_{i}(g_{j})=
        \begin{cases}
            u'_{i}(g_{j}), & \text{for $i\in\{1,\dots,n'\}$ and $j\in\{1,\dots,m'\}$,}\\
            C/(T-2), & \text{for $i\in\{1,\dots,n'\}$ and $j=i+n'+m'$,}\\
            \kappa/(T-2), & \text{for $i\in\{n'+1,\dots,n'+m'\}$ and $j=i-n'$,}\\
            \kappa/2, & \text{for $i\in\{n'+m'+1,\dots,2n'+m'\}$ and $j=i$,}\\
            0, & \text{otherwise.}
        \end{cases}
    \end{equation*}

    Suppose $(\calI',\kappa')$ is a \texttt{YES} instance and let $S'\in\bbS^{2}$ be a solution to this instance. Let $A'$ be the allocation associated to $S'$. Consider the allocation $A$ with
    \begin{equation*}
        A_{ij}=
        \begin{cases}
            A'_{ij}, & \text{for $i\in\{1,\dots,n'\}$ and $j\in\{1,\dots,m'\}$,}\\
            T-2, & \text{for $i\in\{1,\dots,n'\}$ and $j=i+n'+m'$,}\\
            T-2, & \text{for $i\in\{n'+1,\dots,n'+m'\}$ and $j=i-n'$,}\\
            2, & \text{for $i\in\{n'+m'+1,\dots,2n'+m'\}$ and $j=i$,}\\
            0, & \text{otherwise.}
        \end{cases}
    \end{equation*}
    Since the sum of each rows and each columns of $A'$ is at most $2$, it is straightforward to verify that the sum of each rows and each columns of $A$ is at most $T$. Under $A$, agent $i\in\{1,\dots,n'\}$ receives
    \begin{equation*}
        v_{i}(A)=v'_{i}(A')+(T-2)\cdot\frac{C}{T-2}\geq \kappa'+C=\kappa,
    \end{equation*}
    agent $i\in\{n'+1,\dots,n'+m'\}$ receives
    \begin{equation*}
        v_{i}(A)=(T-2)\cdot\frac{\kappa}{T-2}=\kappa,
    \end{equation*}
    and agent $i\in\{n'+m'+1,\dots,2n'+m'\}$ receives
    \begin{equation*}
        v_{i}(A)=2\cdot\frac{\kappa}{2}=\kappa.
    \end{equation*}
    By \Cref{lem:alloc-to-seq}, there exist a sequence $S\in\bbS^{T}$ such that $v_{i}^{T}(S)\geq v_{i}(A)\geq\kappa$. Thus, $(\calI,\kappa)$ is a \texttt{YES} instance.
    
    Suppose $(\calI,\kappa)$ is a \texttt{YES} instance and let $S\in\bbS^{T}$ be a solution to this instance. Let $A$ be the allocation associated to $S$ and $A'$ be the first $n'$ rows and $m'$ columns of $A$. Observe that since agents $i\in\{n'+1,\dots,n'+m\}$ have value at least $\kappa$, $A_{ij}\geq T-2$ for $j=i-n'$. This implies that the sum of each columns of $A'$ is at most $2$. By a similar argument, we have $A_{ij}\geq 2$ for all agents $i\in\{n'+m'+1,\dots,2n'+m'\}$ and $j=i$. Furthermore, since agent $i\in\{1,\dots,n'\}$ can never reach a value of $C$ by only allocating goods in $\{g_{1},\dots,g_{n'+m'}\}$ to them, they must be allocated to $g_{i+n'+m'}$ for at least $T-2$ times. Since these are the same goods that must be allocated at least twice to agents in $\{n'+m'+1,\dots,2n'+m'\}$, we have $A_{ij}=T-2$ for $j=i+n'+m'$. This implies that the sum of each row in $A'$ is also at most $2$. Under $A'$, agent $i\in\{1,\dots,n'\}$ receives
    \begin{equation*}
        v'_{i}(A')=v_{i}(A)-(T-2)\cdot\frac{C}{T-2}\geq\kappa-C=\kappa'.
    \end{equation*}
    By \Cref{lem:alloc-to-seq}, there exist a sequence $S'\in\bbS^{2}$ such that ${v'_{i}}^{2}(S')\geq v'_{i}(A')\geq\kappa'$. Thus, $(\calI',\kappa')$ is a \texttt{YES} instance.

    We now show that \textsc{ERM} is \APX-hard for $T=2$ by showing that if there exists a $(2-\epsilon)$-approx algorithm for \textsc{ERM} (for any small $\epsilon > 0$), then we can determine if there is a sequence $S$ for the instance $\calI$ in our reduction such that $b^{T}(S) \geq 1$. 
    Suppose that $\OPT \geq 1$. A $(2-\epsilon)$-approx algorithm would find a sequence $S$ such that $b^{T}(S)\geq \OPT/(2-\epsilon) \geq 1/(2 -\epsilon) > 0.5$. 
    Now, note that as $u_i(g) \in \{0,0.5,1\}$, for all $i \in N$ and $g \in G$, if $b^{T}(S) > 0.5$  for any sequence $S$, then $b^{T}(S) \geq 1$. Hence,  a $(2-\epsilon)$-approx algorithm allows us to determine if there exists a sequence $S$ for the instance $\calI$ in our reduction such that $b^{T}(S) \geq 1$.

    Lastly, we note that we can easily extend this \APX-hardness result from $T = 2$ to any $T\geq 3$ by adding dummy agents.
    Specifically, for each good, introduce $T -2$ dummy agents that only have non-zero utility for that good and zero for other goods. 
    Hence, to satisfy these dummy agents (so they have non-zero utility), only two ``instances'' of any good can be given to non-dummy agents. 
    Formally, let the set of dummy agents agents be $\{{a_{j,k}} \mid j \in [m], k \in [T-2]\}$. 
    For this dummy set of agents, set 
    \begin{equation*}
        u_{a_{j,k}}(g_{j'})=\begin{cases}
            2 \cdot \max_{g \in G} u_1(g), &\text{if $j=j'$},\\
            0,&\text{otherwise,}
        \end{cases}
    \end{equation*}
    for all $j,j'\in [m]$ and $k \in [T-2]$.
\end{proof}

\SCred*
\begin{proof}
     Let $\Phi=(N,G,(v_{i})_{i\in N})$ be an instance of the Santa Claus problem with additive valuations. 
     Define $\OPT$ to be the maximum egalitarian welfare across all allocations:
     \begin{equation*}
         \OPT = \max_{A} \min_{i \in [n]} v_i(A_i).
     \end{equation*}
     Now, we will construct an instance $\calI$ of \textsc{ERM} such that $\OPT(T)$ for $\calI$ is equal to  $\OPT$ and that for every sequence $S$ such that $b^{T}(S) >0$, $S$ can be easily mapped to an allocation of goods $A = (A_1,\dots, A_n)$ for $\Phi$ such that $\min_{i} v_i(A_i) = b^{T}(S)$. 

     Our instance of ERM has $m$ timesteps, $m$ goods and $n + m \times(m-1)$ agents. 
     We split the agents into two sets $N_1 = \{a_1, \dots, a_n\}$ (representing the agents in $\Phi$) and $N_2 =\{a_{j,k} \mid j \in [m], k \in [m-1]\}$ (representing dummy agents). 
     We construct the utilities as follows: for agents $a_i \in N_1$, let $u_{a_i}(g_j) = v_i(g_j)$, and for agents $a_{j,k} \in N_2$, let 
     \begin{equation*}
         u_{a_{j,k}}(g_{j'}) = \begin{cases}
             v_1(G),&\text{if $j'=j$}\\
             0,&\text{otherwise.}
         \end{cases}
     \end{equation*}

    We now show that for all $\kappa > 0$ there is a sequence $S$  for $\calI$ with bottleneck value $b^{T}(S) \geq \kappa$ if and only if there is an allocation $A$ for $\Phi$ such that $\min_{i \in [n]} v_i(A_i) \geq \kappa$.

    $(\Leftarrow)$ Suppose there is an allocation $A$ for $\Phi$ such that $\min v_i(A_i) \geq \kappa$. 
    Then, we can construct allocation $A'$ for $\calI$ as follows: for $a_i \in N_1$, $A'_{a_i} = A_i$ and for $a_{j,k} \in N_2$, $A'_{a_{j,k}} = \{j\}$. 
    We note that $u_{a_{j,k}}(\{j\}) \geq v_1(G) > \kappa$. 
    We further note that all goods are allocated exactly $m$ times and all agents are allocated at most $m$ goods. 
    Hence,  by \Cref{lem:alloc-to-seq}, there exist a sequence $S\in\bbS^{T}$ such that $v_{i}^{T}(S)\geq v_{i}(A')\geq\kappa$.

    $(\Rightarrow)$ Suppose there is a sequence $S$ for $\calI$ with bottleneck value $b^{T}(S) \geq \kappa$. 
    Then, let $A'$ be the allocation associated with the sequence. We now construct an allocation $A$ for $\Phi$ as follows: for $i \in [n]$, $A_i = A'_{a_i}$. This construction ensures that $v_i(A_i) = u_{a_i}(A_{a_i}') \geq \kappa$. 
    We are now only left to prove that for all goods $g_j \in [m]$, $g_j$ was allocated at most once to agents in $N_1$ (i.e., $\sum_{a_i \in N_1} A'_{a_i,j} \leq 1$).  As $b^{T}(S) > 0$, all agents receive at least one good that they have non-negative utility for. 
    Thus, for all $a_{j,k} \in N_2$, they must be allocated the good $g_j$ at least once. 
    Hence, as $g_j$ must be allocated at least $m-1$ times to agents in $N_2$,  $g_j$ was allocated only once to agents in $N_1$.

    Hence, $\OPT(T)$ for $\calI$ is equal to $\OPT$, and for every sequence $S$ such that $b^{T}(S) > 0$ can be easily mapped to an allocation of goods $A$ for $\Phi$ such that $\min_{i \in [n]} v_i(A_i) \geq b^{T}(S)$. 
    Thus, there is a $c$-approx algorithm for ERM only if there is a $c$-approx algorithm for the Santa Claus problem with additive valuations.
\end{proof} 

\section{Approximate Algorithm for Optimality (\Cref{subsec:approxalgo})}

\approxgeneral*

\begin{proof}
    Consider the allocation $A$ in which $A_{ij}=\lfloor TB_{ij}\rfloor$ for all $i\in N$ and $g_{j}\in G$. Note that for each $g_{j}\in G$, we have
    \begin{equation*}
        \sum_{\mathclap{i\in N}}A_{ij}
        =\sum_{\mathclap{i\in N}}\,\lfloor TB_{ij}\rfloor
        \leq\sum_{\mathclap{i\in N}}TB_{ij}
        =T,
    \end{equation*}
    and similarly, for each $i\in N$, we have
    \begin{equation*}
        \sum_{\mathclap{g_{j}\in G}}A_{ij}
        =\sum_{\mathclap{g_{j}\in G}}\,\lfloor TB_{ij}\rfloor
        \leq\sum_{\mathclap{g_{j}\in G}}TB_{ij}
        =T.
    \end{equation*}
    By \Cref{lem:alloc-to-seq}, there exist a sequence $S$ over $T$ rounds composed of at most $O(m^{2})$ unique matchings such that $v_{i}^{T}(S)\geq v_{i}(A)$. Then, for any agent $i\in N$, we have
    \begin{align*}
        v_{i}^{T}(S)
        \geq v_{i}(A)
        &\geq\sum_{\mathclap{g_{j}\in G}}u_{i}(g_{j})\lfloor TB_{ij}\rfloor\\
        &\geq\sum_{\mathclap{g_{j}\in G}}u_{i}(g_{j})\cdot(TB_{ij}-1)\\
        &=\sum_{\mathclap{g_{j}\in G}}TB_{ij}u_{i}(g_{j})-\sum_{\mathclap{g_{j}\in G}}u_{i}(g_{j})\\
        &\geq T b-m\cdot\max_{\mathclap{g_{j}\in G}}u_{i}(g_{j})\\
        &\geq\OPT(T)-m\cdot\max_{\mathclap{g_{j}\in G}}u_{i}(g_{j}).
    \end{align*}
    Let $k\in N$ be a bottleneck agent of sequence $S$ at round $T$ so that $b^{T}(S)=v_{k}^{T}(S)$. Then, we have
    \begin{align*}
        b^{T}(S)
        &\geq\OPT(T)-m\cdot\max_{\mathclap{g_{j}\in G}}u_{k}(g_{j})\\
        &\geq\OPT(T)-m\cdot\max_{\mathclap{i\in N}}\max_{\mathclap{g_{j}\in G}}u_{i}(g_{j}).\qedhere
    \end{align*}

\end{proof}

\section{FPT Algorithms for Optimality (\Cref{subsec:fptalgo})}
\label{sec:fptalgo}

In this section, we consider another approach to dealing with the computational intractability. Our goal is to develop a \emph{fixed parameter tractable} (FPT) when the number of agents is a fixed parameter, i.e., there exists an algorithm that can compute an optimal sequence in polynomial-time when $n$ is a constant. This provides a practical solution for small-group matching (e.g., crowdsourcing platforms divide workers into subgroups tailored to specific categories of tasks).

To build up to this result, let us first consider the easier case where the number of goods $m$ is the fixed parameter.

\begin{theorem}
    Given an instance $\instfull$, \textsc{ERM} is FPT with respect to $m$.
\end{theorem}
\begin{proof}
    Since the order of the matchings in a sequence does not affect the values at round $T$, our goal is to determine the number of times each matching should be chosen to achieve the highest bottleneck value. Let $\mathcal{M}$ be the set of all possible matchings. For each $M\in\mathcal{M}$, let $X_{M}$ be the number of times $M$ should be chosen in the sequence. We can formulate the optimization problem as the following integer linear program:
    \begin{subequations}
        \begin{alignat}{2}
        \maximize_{b,X} \quad & b\tag{P2}\label{eqn:fpt-general-ilp}\\
        \subjectto \quad & \sum_{\mathclap{M\in\mathcal{M}}}X_{M}u_{i}(M(i))\geq b, && \quad\forall i\in N,\nonumber\\
        & \sum_{\mathclap{M\in\mathcal{M}}}X_{M}=T,\nonumber \\
        & X_{M} \geq 0, && \quad\forall M\in\mathcal{M}.\nonumber
        \end{alignat}
    \end{subequations}
    Since the number of variables for (\ref{eqn:fpt-general-ilp}) is at most $m!+1$, we obtain the FPT result using Lenstra's theorem~\cite{lenstra1983fptilp}.
\end{proof}


In order to extend the result for that the algorithm is FPT with respect to $n$, we need to reduce the number of variables in (\ref{eqn:fpt-general-ilp}) to a function of $n$ and not $m$. This reduction is accomplished through two observations. First, it is always possible to construct an optimal sequence consisting solely of Pareto optimal matchings, by replacing any non-Pareto optimal matching with one that strongly Pareto dominates it. Second, there are at most $n!$ unique (up to its valuation profile) Pareto optimal matchings. Together, these observations allow us to consider a smaller set of matchings $\mathcal{M}$ in (\ref{eqn:fpt-general-ilp}), with size at most $n!$, thereby achieving our desired result. We now prove the second observation by characterizing Pareto optimal matchings in terms of permutations of agents.
We note that this result may be of independent interest, especially on the topic of \emph{house allocation} \citep{abdulkadirouglu1998random,abraham2004pareto,atila1999house,ChooLiSuTeZh24,gan2019envy,hylland1979positions,zhou1990house}.

\paretocharacterization*

\begin{proof}
    If a matching $M$ is $\pi$-optimal for some permutation $\pi$, then $M$ is clearly Pareto optimal: no agent can improve without harming someone with higher priority in $\pi$. Now suppose for contradiction that $M$ is Pareto optimal but not $\pi$-optimal for any permutation $\pi$.
    
    We define the envy\textsuperscript{+} graph of $M$ as a directed graph where each vertex corresponds to an agent. There is an edge from agent $i$ to $i'$ if there exists a sequence of $p\geq 2$ distinct agents $(i_{1},\dots,i_{p})$, with $i_{1}=i$ and $i_{p}=i'$, such that agent $i_{1}$ strictly envies $i_{2}$, meaning $u_{i_{1}}(M(i_{1}))<u_{i_{1}}(M(i_{2}))$, and for $r=2,\dots,p-1$, agent $i_{r}$ weakly envies agent $i_{r+1}$, that is, $u_{i_{r}}(M(i_{r}))\leq u_{i_{r}}(M(i_{r+1}))$.
    
    If this graph contains a cycle, we could perform a cyclic exchange among the agents to strictly improve at least one agent's value without hurting others, contradicting the Pareto optimality of $M$. Hence, the envy\textsuperscript{+} graph must be acyclic. We can therefore define $\pi$ to be the topological sort of this graph in reverse dependency order, so that agents earlier in the ordering are not envied by those that come later. In particular, if $\pi(i)<\pi(i')$, then $(i,i')$ is not an edge in the envy\textsuperscript{+} graph.
        
    Let $M_{*}$ be a $\pi$-optimal matching. We claim that there exists an agent $a\in N$ such that (1) $u_{a}(M(a))<u_{a}(M_{*}(a))$, and (2) for all agent $i\in N$ with $\pi(i)<\pi(a)$, we have $u_{i}(M(i))\leq u_{i}(M_{*}(i))$. To see this, recall that by definition of $\pi$-optimality applied to $M_{*}$, any matching, including $M$, must satisfy one of the following:
    \begin{itemize}
        \item For all agent $i\in N$, $u_{i}(M(i))\leq u_{i}(M_{*}(i))$; or
        \item There exists an agent $i\in N$ such that $u_{i}(M(i))>u_{i}(M_{*}(i))$, but there exists another agent $i'\in N$ with $\pi(i')<\pi(i)$ such that $u_{i'}(M(i'))<u_{i'}(M_{*}(i'))$.
    \end{itemize}
    
    In the first case, if all agents are indifferent between $M$ and $M_{*}$, then $M$ is itself $\pi$-optimal, contradicting our assumption. Hence, there must exist some agent $a\in N$ who is strictly better off under $M_{*}$, that is, $u_{a}(M(a))<u_{a}(M_{*}(a))$, thus proving property (1). Since $u_{i}(M(i))\leq u_{i}(M_{*}(i))$ for all $i$, property (2) follows immediately.

    In the second case, let $i\in N$ be the agent with the highest priority under $\pi$ such that $u_{i}(M(i))>u_{i}(M_{*}(i))$. By definition, there must exist an agent $i'\in N$ with $\pi(i')<\pi(i)$ such that $u_{i'}(M(i'))<u_{i'}(M_{*}(i'))$. Setting $a=i'$ proves property (1). Since all agent with higher priority than $a$ also outrank $i$, and $i$ is the highest-priority agent property who strictly prefers $M$ over $M_{*}$, property (2) follows. This completes the claim, allowing us to proceed with the main argument.
    
    To analyze the structure of $M_{*}$ relative to $M$, we define the envy path graph as a directed graph whose vertices correspond to agents. There is an edge from agent $i$ to agent $i'\neq i$ if, under $M_{*}$, agent $i$ receives the good that agent $i'$ was assigned in $M$, that is, $M_{*}(i)=M(i')$.
    
    Consider the traversal of the envy path graph starting from agent $a$. Since each good is matched to at most one agent in both $M$ and $M_{*}$, each vertex in the graph has at most one incoming and one outgoing edge. As a result, the traversal either enters a cycle that includes agent $a$, or eventually terminates at a vertex with no outgoing edge, forming a simple path. In the former case, we obtain a cycle $(a=a_{1},\dots,a_{q},a_{1})$; and in the latter, we obtain a path $(a=a_{1},\dots,a_{q})$, where it is possible that $q=1$ if there is no outgoing edge from $a$.

    We now show, by induction on the position $r=1,\dots,q-1$ along the traversal path, that (i) $u_{a_{r}}(M(a_{r}))\leq u_{a_{r}}(M(a_{r+1}))$, and (ii) $\pi(a_{r+1})<\pi(a)$. Intuitively, this means that each agent along the path or cycle weakly prefers the good assigned to the next agent, and all agents in the sequence have higher priority than $a$ under $\pi$.
    
    For the base case $r=1$, observe that agent $a=a_{1}$ strictly prefers their assignment in $M_{*}$ over $M$, and by construction, $M_{*}(a_{1})=M(a_{2})$. Therefore, we have
    \begin{equation*}
        u_{a_{1}}(M(a_{1}))=u_{a}(M(a))<u_{a}(M_{*}(a))=u_{a_{1}}(M_{*}(a_{1}))=u_{a_{1}}(M(a_{2})).
    \end{equation*}
    This confirms the first property, showing that agent $a_{1}$ envies $a_{2}$ in $M$. As such, there is an edge from $a_{1}$ to $a_{2}$ in the envy\textsuperscript{+} graph. Since $\pi$ is defined as a reverse topological ordering of this graph, we conclude that $\pi(a_{2})<\pi(a_{1})=\pi(a)$, thereby establishing the second property.

    For the inductive step, assume that both properties hold for all indices up to $r=k-1$. Since $\pi(a_{k})<\pi(a)$ and $M_{*}$ is a $\pi$-optimal matching, it follows that
    \begin{equation*}
        u_{a_{k}}(M(a_{k}))\leq u_{a_{k}}(M_{*}(a_{k}))=u_{a_{k}}(M(a_{k+1})),
    \end{equation*}
    which verifies the first property for $r=k$. To verify the second property, recall that $u_{a_{1}}(M(a_{1}))<u_{a_{1}}(M(a_{2}))$, and for all $r=2,\dots,k$, $u_{a_{r}}(M(a_{r}))\leq u_{a_{r}}(M(a_{r+1}))$. By definition, there must be an edge between $a_{1}$ and $a_{k+1}$ in the envy\textsuperscript{+} graph, which implies that $\pi(a_{k+1})<\pi(a_{1})=\pi(a)$, completing the inductive argument.

    Thus, we return to the two cases of the traversal: either a cycle $(a=a_{1},\dots,a_{q},a_{1})$ or a path $(a=a_{1},\dots,a_{q})$. If the cyclic case, we perform a cyclic exchange where each agent $a_{i}$ receives $M(a_{i+1})$ for all $i\in[q-1]$, and agent $a_{q}$ receives $M(a_{1})$. Since each edge in the envy path graph represents weak preference, every agent in the cycle weakly prefers their new assignment, and agent $a=a_{1}$ strictly prefers $M(a_{2})$ over $M(a_{1})$. This yields a matching that strictly Pareto dominates $M$, contradicting its assumed optimality.

    In the path case, $M_{*}(a_{q})$ must be a good that is unassigned in $M$. We construct a new matching by assigning each agent $a_{i}$ to $M(a_{i+1})$ for $i\in[q-1]$, and assigning agent $a_{q}$ to $M_{*}(a_{q})$. Again, each $a_{i}$ weakly prefers their new good, and $a=a_{1}$ strictly prefers $M(a_{2})$ over $M(a_{1})$. Furthermore, since $\pi(a_{q})<\pi(a)$, the $\pi$-optimality of $M_{*}$ implies that $u_{a_{q}}(M(a_{q}))\leq u_{a_{q}}(M_{*}(a_{q}))$, so agent $a_{q}$ weakly prefers their new good as well. This matching also strictly Pareto dominates $M$, again contradicting its assumed optimality. Thus, $M$ is $\pi$-optimal for some $\pi$.
\end{proof}

Now, note that in an optimal sequence, we can swap a matching with one that has an identical valuation profile without affecting the sequence's optimality. Thus, it is sufficient to compute just one $\pi$-optimal matching for all permutations $\pi$ and consider only sequences construct from these matching. To ensure the resulting algorithm is FPT with respect to $n$, we need to show that it is efficient in all other parameters to compute a $\pi$-optimal matching for each $\pi$.

\begin{restatable}{lemma}{pimatching}
    \label{lem:pi-matching}
    Given a permutation of agents $\pi$, we can find a $\pi$-optimal matching in polynomial time.
\end{restatable}

\begin{proof}
    Denote the rank $r(i,g_{j})$ of good $g_{j}$ for agent $i$ to be the number of goods (inclusive of $g_{j}$) that are valued at most as highly as $g_{j}$. More formally,
    \begin{equation*}
        r(i,g_{j})=\sum_{g\in G}\mathbb{I}[u_{i}(g)\leq u_{i}(g_{j})]
    \end{equation*}
    where $\mathbb{I}[\cdot]$ is the indicator function, which equals $1$ if the condition inside is true and $0$ otherwise. Observe that if two goods have the same value $u_{i}(g_{j})=u_{i}(g)$, then they have the same rank, that is, $r(i,g_{j})=r(i,g)$.

    Construct a complete bipartite graph $H=(N,G,E)$ where the weight of $(i,g)\in E$ is
    \begin{equation*}
        w(i,g)=
        \begin{cases}
            r(i,g)\cdot m^{n-\pi(i)}, & \text{if $i \leq n$,}\\
            0, & \text{otherwise.}\\
        \end{cases}
    \end{equation*}
    Consider the maximum weight matching $M$ in $H$. We claim that $M$ is $\pi$-optimal.
    
    Suppose that $M$ is not $\pi$-optimal. Then there must exists some $M_{0}$ with an agent $i\in N$ that has $u_{i}(M_{0}(i))>u_{i}(M(i))$, which implies $r(i,M_{0}(i))>r(i,M(i))$. Moreover, all agents $i'\in N$ with $\pi(i')<\pi(i)$ must also satisfy $u_{i'}(M_{0}(i'))\geq u_{i'}(M(i'))$, which implies $r(i',M_{0}(i'))\geq r(i',M(i'))$. Partition $N$ into
    \begin{align*}
        N_{1}&=\{i'\in N\mid\pi(i')\leq\pi(i)\},\\
        N_{2}&=\{i'\in N\mid\pi(i')>\pi(i)\}.
    \end{align*}
    Then, we have
    \begin{equation*}
        \sum_{\mathclap{i'\in N_{1}}}w(i',M_{0}(i'))-w(i',M(i'))
        =\sum_{\mathclap{i'\in N_{1}}}\,\bigl(r(i',M_{0}(i'))-r(i',M(i'))\bigr)\cdot m^{n-\pi(i')}
        \geq m^{n-\pi(i)}.
    \end{equation*}
    Furthermore, we also have
    \begin{align*}
        \sum_{\mathclap{i'\in N_{2}}}w(i',M(i'))-w(i',M_{0}(i'))
        &=\sum_{\mathclap{i'\in N_{2}}}\,\bigl(r(i',M(i'))-r(i',M_{0}(i'))\bigr)\cdot m^{n-\pi(i')}\\
        &\leq(m-1)\cdot m^{n-(\pi(i)+1)}+\dots+(m-1)\cdot m^{n-n}\\
        &=m^{n-\pi(i)}-1.
    \end{align*}
    Thus, the weight of matching $M_{0}$ is greater than the weight of the matching $M$ and we have a contradiction.
\end{proof}

Then, we propose \Cref{alg:fpt-general} that gives us our desired result, as follows. 

\begin{algorithm}[t]
    \caption{FPT algorithm for optimal sequence}
    \label{alg:fpt-general}
    \textbf{Input:} an instance $\calI=\instfull$
    \begin{algorithmic}[1]
        \STATE let $\mathcal{M}$ be an empty set
        \FORALL{permutation $\pi$}
            \STATE compute $\pi$-optimal sequence $M$ using \Cref{lem:pi-matching}
            \STATE add $M$ to $\mathcal{M}$
        \ENDFOR
        \STATE construct integer linear program (\ref{eqn:fpt-general-ilp}) based on $\mathcal{M}$
        \STATE let $X$ be the solution to (\ref{eqn:fpt-general-ilp})
        \STATE construct a sequence $S$ with $X_{M}$ copies of $M\in\mathcal{M}$
        \STATE \textbf{return} $S$
    \end{algorithmic}
\end{algorithm}


\fptgeneral*

\begin{proof}
The correctness of \Cref{alg:fpt-general} follows immediately from our discussion. We will now show that it is FPT with respect to $n$. Let $\mathcal{M}$ be the set of $\pi$-optimal matching for all permutations $\pi$. By \Cref{lem:pi-matching}, the loop in \Cref{alg:fpt-general} computes $\mathcal{M}$ in time $\mathrm{poly}(n,m)\cdot n!$. Furthermore, since the number of variables for (\ref{eqn:fpt-general-ilp}) is at most $n!+1$, we obtain the FPT result using Lenstra’s theorem \cite{lenstra1983fptilp}.
\end{proof}

\section{Omitted Proofs in \Cref{sec:anytime-optimality}}

\anytimetwoagents*

\begin{proof}
    We begin by proving that such a sequence always exists, and then demonstrate how to construct it in polynomial time.
    
    \paragraph{Proof of existence.} Consider the case of $m=2$. We will prove by induction in $T$. The statement is obviously true for $T=1$. Suppose the statement is true for $T-1$. Let $\bbS_{1:T-1}=\{S\in\bbS^{T}\mid b^{t}(S)=\OPT(t),\,\forall t<T\}$ be the set of sequences that are anytime optimal up to round $T-1$ and $\bbS_{T}=\{S\in\bbS^{T}\mid b^{T}(S)=\OPT(T)\}$ be the set of sequences that are optimal at round $T$. By the inductive hypothesis, $\bbS_{1:T-1}$ is nonempty. We want to show that $\bbS_{1:T-1}\cap\bbS_{T}\neq\emptyset$.
    
    Suppose, for sake of contradiction, that $\bbS_{1:T-1}\cap\bbS_{T}=\emptyset$. For each pair of sequences $S_{1}=(M_{1}^{1},\dots,M_{1}^{T})\in\bbS_{1:T-1}$ and $S_{2}=(M_{2}^{1},\dots,M_{2}^{T})\in\bbS_{T}$, there is some round $s=\min\{t\in[T]\,|\,M_{1}^{t}\neq M_{2}^{t}\}$ such that the sequences first deviate. Choose $S_{1}$ and $S_{2}$ such that the first deviated round $s$ is maximized.
    
    We claim that for all rounds $t\in\{s,\dots,T\}$, the matching $M_{2}^{t}=M_{2}^{s}$. If there is a round $t$ in which $M_{2}^{t}\neq M_{2}^{s}$, then we can swap these two matchings to obtain a new sequence of matchings $S$. Since $S$ is a rearrangement of $S_{2}$, we know $S\in\bbS_{T}$. Furthermore, since there are only two types of matching for $n=m=2$, the $s$-th matching for $S$ is $M_{1}^{s}$. This implies that the first deviated round between $S_{1}\in\bbS_{1:T-1}$ and $S\in\bbS_{T}$ is greater than $s$, which contradicts the assumption that $s$ is the maximum first deviated round.
    
    Consider the sequence of matchings
    \begin{equation*}
        S_{0}=(M_{2}^{1},\dots,M_{2}^{s-1},M_{1}^{s},M_{2}^{s+1},\dots,M_{2}^{T})
    \end{equation*}
    that is constructed by exchanging the matching of $S_{2}$ at round $s$ to $M_{1}^{s}$. Without loss of generality, we assume that agent $1$ is the bottleneck agent for $S_{2}$ at round $s$.

    If $\min\{u_{1}(g_{1}),u_{1}(g_{2})\}>\max\{u_{2}(g_{1}),u_{2}(g_{2})\}$, then by choosing the good that maximizes the value for agent $2$ for all rounds, we will obtain an anytime optimal sequence of matchings up till round $T$. As such, we only have to consider $\min\{u_{1}(g_{1}),u_{1}(g_{2})\}\leq\max\{u_{2}(g_{1}),u_{2}(g_{2})\}$.
    
    Recall that $S_{0}$ is optimal at round $s$. Since agent $1$ is the bottleneck agent for $S_{2}$ at round $s$, we have $v_{1}^{s}(S_{2})\leq v_{1}^{s}(S_{0})$, which implies that $\min\{u_{1}(g_{1}),u_{1}(g_{2})\}=u_{1}(M_{2}^{s}(1))$. This also implies that $\max\{u_{2}(g_{1}),u_{2}(g_{2})\}=u_{2}(M_{2}^{s}(2))$; otherwise, $M_{4}^{s}$ will weakly Pareto dominates $M_{2}^{s}$, and choosing $M_{4}^{s}$ at every round produces an anytime optimal sequence of matchings up till round $T$. Thus, $v_{1}(M_{2}^{s}(1))\leq v_{2}(M_{2}^{s}(2))$.
    
    To reach a contradiction, we want to show that both $v_{1}^{T}(S_{0})$ and $v_{2}^{T}(S_{0})$ is at least $v_{1}^{T}(S_{2})$, since we will have $b^{T}(S_{0})\geq b^{T}(S_{2})=\OPT(T)$, which implies that $S_{0}\in\bbS_{T}$. Since the first round of deviation between $S_{1}\in\bbS_{1:T-1}$ and $S_{0}\in\bbS_{T}$ is greater than $s$, this leads to a contradiction.
    
    It is straightforward to show that $v_{1}^{T}(S_{0})\geq v_{1}^{T}(S_{2})$:
    \begin{equation*}
        v_{1}^{T}(S_{0})-v_{1}^{T}(S_{2})
        =v_{1}^{s}(S_{0})-v_{1}^{s}(S_{2})
        \geq 0.
    \end{equation*}
    Let us now show that $v_{2}^{T}(S_{0})\geq v_{1}^{T}(S_{2})$. Since $S_{0}$ is optimal at round $s$ and agent $1$ is the bottleneck agent for $S_{2}$ at round $s$, we have $v_{2}^{s}(S_{0})\geq v_{1}^{s}(S_{2})$. Further recall that $u_{1}(M_{2}^{s}(1))\leq u_{2}(M_{2}^{s}(2))$ and $M_{2}^{t}=M_{2}^{s}$ for all rounds $t\in\{s,\dots,T\}$. As such, we have
    \begin{align*}
        &\ v_{2}^{T}(S_{0})-v_{1}^{T}(S_{2})\\
        =&\ v_{2}^{s}(S_{0})+(T-s)\cdot u_{2}(M_{2}^{s}(2))-v_{1}^{s}(S_{2})-(T-s)\cdot u_{1}(M_{2}^{s}(1))\\
        \geq&\ (T-s)\cdot u_{2}(M_{2}^{s}(2))-(T-s)\cdot u_{1}(M_{2}^{s}(1))\\
        =&\ (T-s)\cdot(u_{2}(M_{2}^{s}(2))-u_{1}(M_{2}^{s}(1)))\\
        \geq&\ 0.
    \end{align*}
    Since both $v_{1}^{T}(S_{0})$ and $v_{2}^{T}(S_{0})$ is at least $v_{1}^{T}(S_{2})$, we reached a contradiction.
    
    To extend the proof to general $m$, if each agent's most valued good differs, then it is optimal to match each agent to their most valued good in every round. As such, we only need to consider the case where both agents have the same most valued good.
    
    Suppose that the most valued good of both agents is $g_{j_{0}}$, and let the next most valued good of agent $i$ be $g_{j_{i}}$. Observe that every possible matchings are weakly dominated by either $M_{1}=(g_{j_{1}},g_{j_{0}})$ or $M_{2}=(g_{j_{0}},g_{j_{2}})$. As such, we only need to consider sequences that consist of $M_{1}$ and $M_{2}$. This is equivalent to the case of $m=2$ where the valuation $u'_{i}(g'_{1})=u_{i}(g_{j_{0}})$ and $u'_{i}(g'_{2})=u_{i}(g_{j_{i}})$ for $i\in\{1,2\}$. The existence of an anytime optimal sequence follows from immediately from the proof for $n=2$.
    
    \paragraph{Proof of efficient constructibility.} We first consider the trivial cases. If the agents have different most preferred goods, then we can just assign each agent to their most preferred good for all rounds. Otherwise, let $g_{0}$ be their common most preferred good and $g_{i}$ be the second most preferred good for agent $i\in\{1,2\}$. If there exist an agent $i\in\{1,2\}$ with $u_{i}(g_{0})=u_{i}(g_{i})$, then we can match $g_{i}$ to agent $i$ and $g_{0}$ to the other agent for all rounds. Furthermore, if $u_{1}(g_{0})\leq u_{2}(g_{2})$, then we can just choose the matching $(g_{0},g_{2})$ for all rounds. Similarly, if $u_{2}(g_{0})\leq u_{1}(g_{1})$, then we can just choose the matching $(g_{1},g_{0})$ for all rounds.
    
    Suppose our instance is not one of the trivial cases. Let $M_{1}=(g_{0},g_{2})$ and $M_{2}=(g_{1},g_{0})$. Note that all matching is weakly Pareto dominated by either $M_{1}$ or $M_{2}$. As such, there must exist an anytime optimal sequence that contains only $M_{1}$ and $M_{2}$. We now construct the sequence $S$ greedily and iteratively. Consider the following loop invariant that must be satisfied before the start of iteration $t\in[T]$:
    \begin{quote}
        The sequence $S\in\bbS^{t-1}$ is anytime optimal up till round $t-1$ and there exist an extension of the sequence such that it is anytime optimal up till round $t'\geq t$.
    \end{quote}
    This is satisfied before the start of iteration $t=1$ because $S$ is empty (hence vacuously anytime optimal) and there exist an extension that is anytime optimal.
    
    Suppose the loop invariant is satisfied before the start of iteration $t\in[T]$. We want to extend the sequence with a matching such that the loop invariant is satisfied before the start of iteration $t+1$. By the loop invariant, we know that there exist an extension of the sequence (using only $M_{1}$ and $M_{2}$) that is anytime optimal up till round $t'\geq t$. Let $S_{1}=S\cup M_{1}$ and $S_{2}=S\cup M_{2}$. If $b^{t}(S_{1})>b^{t}(S_{2})$, then extending $S$ to $S{1}$ ensures the loop invariant holds for $t+1$. This holds similar for the case of $b^{t}(S_{1})<b^{t}(S_{2})$.
    
    Suppose that $b^{t}(S_{1})=b^{t}(S_{2})$. Since $u_{i}(g_{0})>u_{1}(g_{i})$ for both $i\in\{1,2\}$, we have $v_{1}^{t}(S_{1})>v_{1}^{t}(S_{2})$ and $v_{2}^{t}(S_{1})<v_{2}^{t}(S_{2})$. 
    Note that agent $1$ (resp. agent $2$) is the unique bottleneck agent for $S_{2}$ (resp. $S_{1}$). To see this, suppose agent $2$ is a bottleneck agent for $S_{2}$, that is, $b^{t}(S_{2})=v_{2}^{t}(S_{2})\leq v_{1}^{t}(S_{2})$. Then, we have $v_{2}^{t}(S_{1})<v_{2}^{t}(S_{2})\leq v_{1}^{t}(S_{2})<v_{1}^{t}(S_{1})$, which implies that agent $2$ is the bottleneck agent for $S_{1}$, that is, $b^{t}(S_{1})=v_{2}^{t}(S_{1})$. This leads to a contradiction because we have 
    \begin{equation*}
        b^{t}(S_{1})=v_{2}^{t}(S_{1})<v_{2}^{t}(S_{2})=b^{t}(S_{2}).
    \end{equation*}
    A similar argument can be used to prove the respective case. These results imply that $v_{2}^{t}(S_{1})=b^{t}(S_{1})=b^{t}(S_{2})=v_{1}^{t}(S_{2})$.

    We will now show that it does not matter if we match $g_{0}$ to any agent in round $t$ because we will have to match $g_{0}$ to the other agent in round $t+1$. Suppose, for sake of contradiction, and without loss of generality, that it is optimal to match $g_{1}$ to agent $1$ for round $t$ and $t+1$. Then, we have
    \begin{align*}
        &\ \smash{\min\{v_{1}^{t-1}(S)+2u_{1}(g_{1}), v_{2}^{t-1}(S)+2u_{2}(g_{0})\}}\\
        =&\ \smash{\min\{v_{1}^{t}(S_{2})+u_{1}(g_{1}), v_{2}^{t}(S_{2})+u_{2}(g_{0})\}}\\
        \leq&\ v_{1}^{t}(S_{2})+u_{1}(g_{1})\\
        =&\ v_{2}^{t}(S_{1})+u_{1}(g_{1})\\
        \leq &\ v_{2}^{t}(S_{1})+\max\{u_{1}(g_{1}),u_{2}(g_{2})\}\\
        \leq&\ v_{2}^{t}(S_{1})+\min\{u_{1}(g_{0}),u_{2}(g_{0})\}\\
        =&\ \smash{\min\{v_{1}^{t-1}(S)+u_{1}(g_{0})+u_{1}(g_{1}), v_{2}^{t-1}(S)+u_{2}(g_{0})+u_{2}(g_{2})\}},
    \end{align*}
    which contradicts optimality at round $t+1$. A similar argument can be used to show that it is not optimal to match $g_{2}$ to agent $2$ for round $t$ and $t+1$. Hence, the only extension left is to either match $g_{0}$ to any agent in round $t$ and to the other agent in round $t+1$. 
\end{proof}

\anytimedontexist*

\begin{proof}
    Consider the following instance with $m=n=3$. For each $i\in N$ and $g_{j}\in G$, let $u_{i}(g_{j})=U_{ij}$, where
    \begin{equation*}
        U=
        \begin{bmatrix}
            5 & 2 & 1 \\
            3 & 3 & 2 \\
            2 & 5 & 1
        \end{bmatrix}.
    \end{equation*}
    Note that $\OPT(1)=2$ and $\OPT(2)=6$. Furthermore, the only way to achieve $\OPT(2)$ is by choosing $M_{1}=(1,2,3)$ and $M_{2}=(3,1,2)$ in any order. As such, the bottleneck value at $t=1$ is $1$, which is not anytime optimal.
\end{proof}

\anytimehardness*

\begin{proof}
    To prove that our problem is \coNP-hard, we will show that the complement of our problem is \NP-hard by reducing from the \textsc{3-partition} problem. In the \textsc{3-partition} problem, we are given a multiset $R=\{a_{1},\dots,a_{3d}\}$ and we need to decide if there exists a partition of $R$ into $d$ triplets such that the sum of all triplets equals to
    \begin{equation*}
        \gamma=\frac{1}{d}\sum_{k=1}^{3d}a_{k}.
    \end{equation*}

    Given a \textsc{3-partition} instance $R$, let $\epsilon<\min\{\gamma/8-1/4,1/6\}$. We will reduce $R$ to an instance $\calI=\instfull$ with $n=m=5d+3$, $T=3$, and $u_{i}(g_{j})=U_{ij}$, where
    \begin{figure}[h]
        \centering
        \begin{tikzpicture}[x=0.7cm,y=-0.7cm]
        \draw[anchor=east] (-0.2,5.5) node {\LARGE $U=$};
        \draw[anchor=west] (11.1,1) node {$d$};
        \draw[anchor=west] (11.1,5) node {$3d$};
        \draw[anchor=west] (11.1,9) node {$d$};
        \draw[anchor=west] (11.1,10.5) node {$3$};
        \draw[anchor=north] (3,11.1) node {$3d$};
        \draw[anchor=north] (7,11.1) node {$d$};
        \draw[anchor=north] (9,11.1) node {$d$};
        \draw[anchor=north] (10.5,11.1) node {$3$};
        \draw (0,0) rectangle (11,11);
        \draw (0,2) -- (11,2);
        \draw (0,8) -- (11,8);
        \draw (0,10) -- (11,10);
        \draw (0,10) -- (11,10);
        \draw (6,0) -- (6,11);
        \draw (8,0) -- (8,11);
        \draw (10,0) -- (10,11);
        \draw (0.5,0.5) -- (5.5,0.5);
        \draw (0.5,1.5) -- (5.5,1.5);
        \draw[dotted,thick] (3,0.5) -- (3,1.5);
        \filldraw[white] (2.7,0.2) rectangle (3.3,0.8);
        \filldraw[white] (2.7,1.2) rectangle (3.3,1.8);
        \draw (3,0.5) node {$a$};
        \draw (3,1.5) node {$a$};
        \draw (7,1) node {\LARGE$2\epsilon$};
        \draw (9,1) node {\LARGE$\frac{\gamma-1}{3}$};
        \draw (10.5,1) node {\LARGE$0$};
        \draw[dotted,thick] (0.5,2.5) -- (5.5,7.5);
        \filldraw[white] (0.2,2.2) rectangle (0.8,2.8);
        \filldraw[white] (5.2,7.2) rectangle (5.8,7.8);
        \draw (0.5,2.5) node {$\nicefrac{\gamma}{2}$};
        \draw (5.5,7.5) node {$\nicefrac{\gamma}{2}$};
        \draw (1.5,6.5) node {\LARGE$2\epsilon$};
        \draw (4.5,3.5) node {\LARGE$2\epsilon$};
        \draw (7,5) node {\LARGE$0$};
        \draw (9,5) node {\LARGE$0$};
        \draw (10.5,5) node {\LARGE$2\epsilon$};
        \draw (3,9) node {\LARGE$0$};
        \draw (7,9) node {\LARGE$\frac{\gamma-1}{3}$};
        \draw (9,9) node {\LARGE$\frac{\gamma}{3}$};
        \draw (10.5,9) node {\LARGE$0$};
        \draw (3,10.5) node {\LARGE$2\epsilon$};
        \draw (7,10.5) node {\LARGE$0$};
        \draw (9,10.5) node {\LARGE$0$};
        \draw (10.5,10.5) node {\LARGE$Z$};
        \draw[anchor=south west] (12,11) node {$Z=\begin{bmatrix}
        \epsilon & \epsilon & \gamma-2\epsilon\\
        \epsilon & \epsilon & \gamma-2\epsilon\\
        \epsilon & \epsilon & \gamma-2\epsilon
        \end{bmatrix}$};
        \end{tikzpicture}
    \end{figure}
    
    Observe that $\OPT(1)\leq 2\epsilon$ and $\OPT(2)\leq 4\epsilon$ since at least one agents in $\{5d+1,5d+2,5d+3\}$ cannot be matched to good $g_{5d+3}$ in the first two rounds. We also have $\OPT(1)\geq 2\epsilon$ and $\OPT(2)\geq 4\epsilon$. This can be achieved by considering the sequence that, for all rounds, matches agent $i\in\{1,\dots,d\}$ to $g_{3d+i}$, agent $d+1,d+2,d+3$ to $g_{5d+1},g_{5d+2},g_{5d+3}$ respectively, agent $i\in\{d+4,\dots,4d\}$ to $g_{i-d}$, agent $i\in\{4d+1,\dots,5d\}$ to $g_{i}$, and agent $5d+1,5d+2,5d+3$ to $g_{1},g_{2},g_{3}$ respectively. Thus, we have $\OPT(1)=2\epsilon$ and $\OPT(2)=4\epsilon$.
    
    Suppose $R$ is a \texttt{YES} instance and $R_{1},\dots,R_{d}$ is a solution to the instance. We claim that $\OPT(3)\geq\gamma$. To see this, consider the allocation $A$ in which
    \begin{equation*}
        A_{ij}=
        \begin{cases}
            1, & \text{if $i\in\{1,\dots,d\}$ and $a_{j}\in R_{i}$,}\\
            2, & \text{if $i\in\{d+1,\dots,4d\}$ and $j=i-d$,}\\
            1, & \text{if $i\in\{d+1,\dots,4d\}$ and $j=3d+\lceil(i-d)/3\rceil$,}\\
            3, & \text{if $i\in\{4d+1,\dots,5d\}$ and $j=i$,}\\
            1, & \text{if $i\in\{5d+1,5d+2,5d+3\}$ and $j\in\{5d+1,5d+2,5d+3\}$}\\
            0, & \text{otherwise.}
        \end{cases}
    \end{equation*}
    By \Cref{lem:alloc-to-seq}, we can convert $A$ into a sequence $S$ of three matchings where each $i\in N$ satisfy $v_{i}^{3}(S)\geq\gamma$.

    For sake of contradiction, suppose $\calI$ is a \texttt{YES} instance and let $S$ be an anytime optimal sequence. Since $S$ is anytime optimal, it is optimal at $t=3$, i.e., $b^{3}(S)=\OPT(3)\geq\gamma$. As such, agent $i\in\{4d+1,\dots,5d\}$ must be matched to $g_{i}$ for all three rounds. Furthermore, agent $i\in\{d+1,\dots,4d\}$ has to be matched to $g_{i-d}$ at least twice; otherwise, $v_{i}^{3}(S)$ would be smaller than $6\epsilon$ or $4\epsilon+\gamma/2$, which are strictly less than $\gamma$ by our choice of $\epsilon$. As such, each good $g_{1},\dots,g_{3d}$ can only be matched to agent $1,\dots,d$ at most once over all rounds. We claim that these goods must be matched to these agents exactly once over all rounds. Suppose, for sake of contradiction, there is some good $g$ that is not matched to any agents in $\{1,\dots,d\}$. Then, we have
    \begin{equation*}
        \sum_{i=1}^{d}v_{i}^{3}(S)
        \leq\Biggl(\sum_{j=1}^{3d}u_{1}(g_{j})\Biggr)-u_{1}(g)+2\epsilon
        <\gamma d-1+\frac{1}{3}
        <\gamma d,
    \end{equation*}
    since $u_{1}(g)$ is a positive integer and $\epsilon<1/6$. As such, at least one agent $i\in\{1,\dots,d\}$ has $v_{i}^{3}(S)<\gamma$, which contradicts to the optimality of $S$. Since $g_{1},\dots,g_{3d}$ and $g_{4d+1},\dots,g_{5d}$ must be allocated three times to agents $1,\dots,5d$, these cannot be matched to any agent $5d+1,5d+2,5d+3$. Thus, at least one agent $i\in\{5d+1,5d+2,5d+3\}$ has $v_{i}^{1}(S)\leq\epsilon<\OPT(1)$, which contradicts our assumption that $S$ is an anytime optimal sequence. Thus, $\calI$ is a \texttt{NO} instance.
    
    Suppose $R$ is a \texttt{NO} instance. We claim that $\OPT(3)\leq\gamma-1$. Suppose, for sake of contradiction, that $\OPT(3)>\gamma-1$. Then, agent $i\in\{4d+1,\dots,5d\}$ must be matched to $g_{i}$ for all three rounds. Furthermore, agent $i\in\{d+1,\dots,4d\}$ has to be matched to $g_{i-d}$ at least twice; otherwise, $v_{i}^{3}(S)$ would be smaller than $6\epsilon$ or $4\epsilon+\gamma/2$, which are strictly less than $\gamma-1$ by our choice of $\epsilon$. Observe that if there is some $g\in\{g_{1},\dots,g_{3d}\}$ that is not matched to any agent in $\{1,\dots,d\}$ for all rounds, then there must exist an agent $i\in\{1,\dots,d\}$ that is matched to $g'\notin\{g_{1},\dots,g_{3d}\}$, and we can strictly improve its valuation by swapping out $g'$ with $g$. Thus, we only need to consider when every goods in $\{1,\dots,3d\}$ are matched to exactly one agent once. However, since $R$ is a \texttt{NO} instance, we know that for all partition of $\{g_{1},\dots,g_{3d}\}$ into $d$ triplets $G_{1},\dots,G_{d}$, there must exist some subset $G_{i}$ in which the sum of its valuation is at most $\gamma-1$. Thus, there exist an agent with valuation at most $\gamma-1$.

    Now, we construct an anytime optimal sequence $S$ with $\OPT(1)=2\epsilon$, $\OPT(2)=4\epsilon$, and $\OPT(3)=\gamma-1$. Consider the allocation $A$ in which
    \begin{equation*}
        A_{ij}=
        \begin{cases}
            3, & \text{if $i\in\{1,\dots,d\}$ and $j=4d+i$,}\\
            2, & \text{if $i\in\{d+1,\dots,d+6\}$ and $j=i-d$,}\\
            1, & \text{if $i\in\{d+1,\dots,d+6\}$ and $j=5d+\lceil(i-d)/3\rceil$,}\\
            3, & \text{if $i\in\{d+7,\dots,4d\}$ and $j=i-d$,}\\
            3, & \text{if $i\in\{4d+1,\dots,5d\}$ and $j=i-d$,}\\
            1, & \text{if $i=5d+1$ and $j\in\{1,2,5d+3\}$,}\\
            1, & \text{if $i=5d+2$ and $j\in\{3,4,5d+3\}$,}\\
            1, & \text{if $i=5d+3$ and $j\in\{5,6,5d+3\}$,}\\
            0, & \text{otherwise.}
        \end{cases}
    \end{equation*}
    It is straightforward to verify that $v_{i}(A)\geq\gamma-1$. By \Cref{lem:alloc-to-seq}, we can convert $A$ into a sequence $S$ of three matchings where each $i\in N$ satisfy $v_{i}^{3}(S)\geq\gamma-1$. Furthermore, since all the goods that are matched to each agent has value at least $2\epsilon$, we have $v_{i}^{1}(S)\geq 2\epsilon$ and $v_{i}^{2}(S)\geq 4\epsilon$. This implies that $S$ is an anytime optimal sequence. Thus, $\calI$ is a \texttt{YES} instance.
\end{proof}

\anytimeapprox*

\begin{proof}
    Observe that since (\ref{eqn:approx-general-lp}) has $m^{2}+5m$ inequality constraints and $m^{2}+1$ variables, $m^{2}+1$ constraints will be tight at a vertex solution, meaning there are at most $5m$ non-zero entries in $B$, which implies that $d\leq 5m$.
    
    Let $n_{kt}$ be the value of $n_{k}$ after round $t$. After each round $t\in[T]$, we claim that our choice of matching $M^{t}$ maintains the invariant $\alpha_{k}t-n_{kt}\leq 1$ for all $k\in[d]$. If the invariant is kept, we have
    \begin{align*}
        v_{i}^{t}(S)
        &=\sum_{k=1}^{d}u_{i}(M_{k}(i))\,n_{kt}\\
        &\geq\sum_{k=1}^{d}u_{i}(M_{k}(i))(\alpha_{k}t-1)\\
        &\geq\sum_{k=1}^{d}u_{i}(M_{k}(i))\,\alpha_{k}t-d\cdot\max_{g\in G}u_{i}(g)\\
        &=\sum_{\mathclap{g_{j}\in G}}u_{i}(g_{j})\,tB_{ij}-d\cdot\max_{g\in G}u_{i}(g)\\
        &\geq tb-d\cdot\max_{g\in G}u_{i}(g)\\
        &\geq\OPT(t)-5m\cdot\max_{g\in G}u_{i}(g),
    \end{align*}
    for all agents $i\in N$, where the fourth line is true because
    \begin{equation*}
        \sum_{k=1}^{d}u_{i}(M_{k}(i))\,\alpha_{k}
        =\sum_{k=1}^{d}\sum_{g_{j}\in G}u_{i}(g_{j})\,\alpha_{k}(M_{k})_{ij}
        =\sum_{\mathclap{g_{j}\in G}}u_{i}(g_{j})\,\Biggl[\sum_{k=1}^{d}\alpha_{k}M_{k}\Biggr]_{ij}
        =\sum_{\mathclap{g_{j}\in G}}u_{i}(g_{j})\,B_{ij}.
    \end{equation*}
    
    We are left to show that the invariant is kept after each round. Let $g_{kt}=(n_{kt}+1)/\alpha_{k}$ for each $k\in[d]$. Suppose, for sake of contradiction, that $\alpha_{k}t-n_{kt}>1$ for some $t\in[T]$ and some $k\in[d]$. By rearranging the terms, we have $t>(n_{kt}+1)/\alpha_{k}=g_{kt}$. 
    
    For all other matching $M_{l}\neq M_{k}$, if $M_{l}$ is not chosen for any round $s\leq t$, then we have $n_{lt}=0<\alpha_{l}t$. Otherwise, suppose that $M_{l}$ is chosen for the $n_{lt}$ time on round $s\leq t$, that is, $n_{lt}=n_{ls}=n_{l,s-1}+1$. Since $M_{l}$ is chosen over $M_{k}$, we must have $g_{l,s-1}\leq g_{k,s-1}$. Then, we have
    \begin{equation*}
        \frac{n_{lt}}{\alpha_{l}}
        =\frac{n_{l,s-1}+1}{\alpha_{l}}
        =g_{l,s-1}
        \leq g_{k,s-1}
        =\frac{n_{k,s-1}+1}{\alpha_{k}}
        \leq\frac{n_{kt}+1}{\alpha_{k}}
        <t
    \end{equation*}
    Thus, we have that $n_{lt}<\alpha_{l}t$ for all $l\in[d]$. Summing across all $n_{lt}$, we have
    \begin{equation*}
        \sum_{l\in[d]}n_{lt}
        <\sum_{l\in[d]}\alpha_{l}t
        =t\sum_{l\in[d]}\alpha_{l}
        =t 
    \end{equation*}
    where the last equality is due to the fact that the sum of the weights is $1$. However, this is a contradiction because we select a matrix at every timestep, and thus $\sum_{l\in[d]}n_{lt}$ has to be $t$.
\end{proof}

\section{Omitted Proofs in \Cref{sec:special-cases}}
We first define circulation with demand.

\begin{definition}[Circulation with demand]
    Let $G=(V,E)$ be a directed graph where each vertex $v\in V$ has a demand $d(v)$. A \emph{circulation with demand} is a function $f:E\to\bbR$ that assigns non-negative value to each edge $(u,v)\in E$ such that
    \begin{equation*}
        \sum_{\mathclap{(u,v)\in E}}f(u,v)-\sum_{\mathclap{(v,w)\in E}}f(v,w)=d(v),\quad\forall v\in V.
    \end{equation*}
\end{definition}

\matchingpolem*

\begin{proof}
    We prove by construction. Take any matching $M$ and let $M'$ be any maximum matching with goods in $G'$. If all agents in $M$ is matched to some good in $G'$, then we are done. Let $N_{0}$ be the set of agents in $M$ that are not matched to goods in $M'$ and let $G_{0}$ be the set of goods in $M'$ that are not matched to agents in $M$.
    Note that each agent $i\in N_{0}$ must be in $M'$; otherwise, we can add $(i,M(i))$ to $M'$ to increase its cardinality, contradicting to the maximality of $M'$.

    We define an $\emph{augmenting graph}$ as follows. The goods $G$ are the vertices of the augmenting graph and $(g,g')\in E$ if there exists an agent that receives $g$ in $M$ and $g'$ in $M'$. Note that each vertex have at most one incoming edge and at most one outgoing edge, and that for each $i\in N_{0}$, we have $(M(i),M'(i))\in E$. Furthermore, this edge cannot be part of a cycle in the augmenting graph because $M(i)$ is not in $G'$, thus it follows a non-cyclic path $P=(M(i)=g^{1},\dots,g^{k})$. We denote $i^{q}$ to be the agent such that $M(i^{q})=g^{q}$ and $M'(i^{q})=g^{q+1}$.
    
    We first process the agents in $N_{0}$ where no agent in $M$ is matched to $g^{k}$ on its path $P$. We modify $M$ by removing $(i^{q},g^{q})$ and adding $(i^{q},g^{q+1})$ for each $q\in[k-1]$. This ensures that each agent in $M$ who is on $P$ is now matched to a good in $G'$, and that $i^{1}$ and $g^{k}$ can be removed from $N_{0}$ and $G_{0}$ respectively.

    We claim that after the processing step, we are done. Specifically, there is no agent in $N_{0}$ where an agent in $M$ is matched to $g^{k}$ on its path $P$. Suppose such an agent $i\in N_{0}$ exist. Since $|N_{0}|=|G_{0}|$, there must be a good $g'\in G_{0}$ and an agent $i'$ in which $M'(i')=g'$. Since no agent is matched to $g'$ in $M$ (by definition of $G_{0}$) and no good is matched to $i'$ in $M$ (because this reduces to the previously processed case, which we assume are all processed), we can add $(i',g')$ to $M$ to increase its cardinality, which contradicts to its maximality.
\end{proof}

\binaryoptimal*

\begin{proof}
    We solve this with binary search on the optimal bottleneck value $b$. Let $b$ be the guess of the optimal bottleneck value. 
    
    Let $M$ be a maximum matching and let $G'$ be the goods allocated in $M$. By Lemma \ref{lem:matching_po}, it is sufficient to consider only matchings formed with these goods at every timestep. Then, if $nb > |G'|T$, we automatically reject as there can be no outcome where all agents are satisfied at $b$ timesteps.

    Otherwise, we will consider the following circulation with demand problem. 

    For each agent $i \in N$, we create a vertex $u_i$ with demand $b$ and for each $g_j \in G'$, we create a vertex $v_j$ with demand $-T$. Then, We add an edge $(u_i,v_j)$ if $u_i(g_j) = 1$. Finally, we create a vertex $i_0$ with demand $|G'|T - nb$ and add an edge between $i_0$ and all nodes $v_j$ for $g_j \in G'$.

    We claim that there is a feasible circulation if and only if there is a sequence of matching with bottleneck value at least $b$.

   $(\Rightarrow)$ Suppose there is a feasible circulation $f$. Note that since all the demands are integer-valued, the resulting circulation is also integer-valued. As such, we can consider the allocation $A$ where $A_{ij} = f(u_i,v_j)$. We note that by our demand constraint, the sum of all rows are $b$ and the sum of all columns is at most $T$. Then, we can add $n - |G'|$ empty columns (that represents `fake' zero-valued goods) and by \Cref{lem:alloc-to-seq}, there exist a sequence of matching $S$ such that $v_{i}^{T}(S)\geq b$, which implies that $b^{T}(S)\geq b$. 

    $(\Leftarrow)$ Suppose there is a sequence of matching $S$ with $b^{T}(S)\geq b$. By Lemma~\ref{lem:matching_po}, we can assume that $S$ is chosen such that for every timestep, $u_{i}(M_{t}(i))=1$ if and only if $M_{t}(i)\in G'$. Then,  let $A$ be the allocation that correspond to that matching. We note that every row sums up to at most $T$. Then the circulation $f(u_{i},v_{j})= A_{ij}$ for $i \in [n], j \in [|G'|]$ and $f(u_{0},v_{j})=T- \sum_{i \in [n]}A_{ij}$ for $j \in [|G'|]$ is feasible.
\end{proof}

\optimaltwogoods*

\begin{proof}
    We solve this with binary search on the optimal bottleneck value $b$. Let $b$ be the guess of the optimal bottleneck value. If $T\cdot\max_{g\in G}u_{i}(g)\leq b$ for some agent $i\in N$, we can immediately reject $b$. Otherwise, we will consider the following circulation with demand problem:
    
    For each good $g\in G$, create a vertex $g$ with demand $T$. For each agent $i\in N$, if agent $i$ is indifferent between $G_{0}$ and $G_{1}$, we create a vertex $i_{0}$ with demand $-T$ and add an edge $(i_{0},g)$ to all $g\in G$. Otherwise, agent $i$ strictly prefers $G_{r}$ for some $r\in\{0,1\}$. Let $g_{r}\in G_{r}$, $g_{1-r}\in G_{1-r}$, and
    \begin{equation*}
        k_{i}=\biggl\lceil\frac{b-T\cdot u_{i}(g_{1-r})}{u_{i}(g_{r})-u_{i}(g_{1-r})}\biggr\rceil
    \end{equation*}
    be the minimum number of rounds that agent $i$ needs to receive goods from $G_{r}$ to achieve at least $b$ in valuation. We create two vertices $i_{r}$ and $i_{1-r}$ with demand $-k_{i}$ and $-(T-k_{i})$ respectively. We then add an edge $(i_{r},g)$ if $g\in G_{r}$ and another edge $(i_{1-r},g)$ for all $g\in G$. Finally, we create a source $s$ with demand $-T(m-n)$ and add an edge $(s,g)$ for all $g\in G$. We claim that there is a feasible circulation if and only if there is a sequence of matching with bottleneck value at least $b$.

    $(\Rightarrow)$ Suppose there is a feasible circulation $f$. Note that since all the demands are integer-valued, the resulting circulation is also integer-valued. As such, we can consider the allocation $A$ where
    \begin{equation*}
        A_{ij}=
        \begin{cases}
            f(i_{0},g_{j}),&\text{if agent $i$ is indifferent between $G_{0}$ and $G_{1}$,}\\
            f(i_{r},g_{j}),&\text{if agent $i$ strictly prefers $G_{r}$ and $g_{j}\in G_{r}$,}\\
            f(i_{0},g_{j})+f(i_{1},g_{j}),&\text{if agent $i$ strictly prefers $G_{r}$ and $g_{j}\in G_{1-r}$,}\\
            0,&\text{otherwise.}
        \end{cases}
    \end{equation*}
    Let $i\in N$. If agent $i$ is indifferent between $G_{0}$ and $G_{1}$, then it trivially satisfy $v_{i}(A)\geq b$. Suppose agent $i$ strictly prefers $G_{r}$. Then, we have
    \begin{align*}
        v_{i}(A)
        &=\sum_{\mathclap{g\in G_{r}}}f(i_{r},g)u_{i}(g)+\sum_{\mathclap{g\in G}}f(i_{1-r},g)u_{i}(g)\\
        &\geq u_{i}(g_{r})\sum_{\mathclap{g\in G_{r}}}f(i_{r},g)+u_{i}(g_{1-r})\sum_{\mathclap{g\in G}}f(i_{1-r},g)\\
        &=k_{i}\cdot u_{i}(g_{r})+(T-k_{i})\cdot u_{i}(g_{1-r})\\
        &\geq b,
    \end{align*}
    where the last inequality holds by our choice of $k_{i}$. Then, by \Cref{lem:alloc-to-seq}, there exist a sequence of matching $S$ such that $v_{i}^{T}(S)\geq b$, which implies that $b^{T}(S)\geq b$.
    
    $(\Leftarrow)$ Suppose there is a sequence of matching $S$ with $b^{T}(S)\geq b$. For each $r\in\{0,1\}$, let $L_{r}$ be a dynamic set that is initialized to contain $T$ of each good $g\in G_{r}$, and define the operation $\texttt{remove}:\{0,1\}\times\bbN\rightarrow\mathcal{P}(G\times\bbN)$ such that $\texttt{remove}(r,c)$ removes $c$ elements from $L_{r}$ and returns the number of times each good from $G_{r}$ is removed.
    
    Let $A$ be an allocation associated with the sequence $S$. We now construct the flow $f$. For $i\in N$, if agent $i$ is indifferent between $G_{0}$ and $G_{1}$, then we set
    \begin{equation*}
        f(i_{0},g_{j})=
        \begin{cases}
            k,&\quad\text{if $g_{j}\in G_{0}$ and $(g_{j},k)\in\texttt{remove}(0,A_{ij})$},\\
            k,&\quad\text{if $g_{j}\in G_{1}$ and $(g_{j},k)\in\texttt{remove}(1,A_{ij})$},\\
            0,&\quad\text{otherwise.}
        \end{cases}
    \end{equation*}
    
    Otherwise, suppose that agent $i$ strictly prefer $G_{r}$ over $G_{1-r}$. Since $v_{i}(A)\geq b$, we know that the number of rounds agent $i$ is allocated goods in $G_{r}$ is at least $k_{i}$. Let
    \begin{equation*}
        A_{ir}=\sum_{\mathclap{g_{j}\in G_{r}}}A_{ij}\quad\text{and}\quad A_{i(1-r)}=\sum_{\mathclap{g_{j}\in G_{1-r}}}A_{ij}.
    \end{equation*}
    We can set
    \begin{equation*}
        f(i_{r'},g_{j})=
        \begin{cases}
            k,&\quad\text{if $r'=r$ and $g_{j}\in G_{r}$ and $(g_{j},k)\in\texttt{remove}(r,k_{i})$},\\
            k,&\quad\text{if $r'=1-r$ and $g_{j}\in G_{r}$ and $(g_{j},k)\in\texttt{remove}(r,A_{ir}-k_{i})$},\\
            k,&\quad\text{if $r'=1-r$ and $g_{j}\in G_{1-r}$ and $(g_{j},k)\in\texttt{remove}(1-r,A_{i(1-r)}-k_{i})$},\\
            0,&\quad\text{otherwise.}
        \end{cases}
    \end{equation*}
    We note that thus far for all vertices corresponding to the agents, the sum of flow outgoing from the node is equal to its demand and for all nodes corresponding to the goods, the flow incoming is at most $T$. Thus, we can set the flow outgoing from node $s$ appropriately. 
\end{proof}


\optimalidenticalhardness*

\begin{proof}
    We will first define the Promise Balanced Partition problem and show it is \NP-hard before proving \Cref{thm:optimal-identical-hardness} is \NP-complete.
    
    \begin{tcolorbox}[title=\textsc{Promise Balanced Partition (PBP)}]
        \textbf{Input}: A list of distinct non-negative integer $E=\{e_{1},\dots,e_{k}\}$. Let $\tau$ be the sum of the elements in $E$. It is guaranteed that for all multisets of size $k/2$, the sum of its elements does not equal $\tau/2$ if it contains duplicates.
        \tcblower
        \textbf{Question}: Is there a partition of $E$ into two equal-sized subsets $E_{1},E_{2}$ such that the sum of the elements in $E_{1}$ equals to the sum of the elements in $E_{2}$?
    \end{tcolorbox}

    Let $\Phi$ be an instance of \textsc{1-in-3-sat}. For each $i\in[n]$, we define two integers $t_{i}$ and $f_{i}$ that encode the assignments of the variable $x_{i}$ to \texttt{TRUE} and \texttt{FALSE}, respectively. In particular, we define
    \begin{equation*}
        t_{i}=(2n+1)^{m+n-i}+\sum_{\mathclap{j:x_{i}\in c_{j}}}(2n+1)^{m-j}
        \quad\text{and}\quad
        f_{i}=(2n+1)^{m+n-i}+\sum_{\mathclap{j:\bar{x}_{i}\in c_{j}}}(2n+1)^{m-j}.
    \end{equation*}
    We additionally define $t_{n+1}=\sum_{j=1}^{m}(2n+1)^{m-j}$ and $f_{n+1}=0$. Let
    \begin{equation*}
        E=\{t_{1},f_{1},\dots,t_{n+1},f_{n+1}\}.
    \end{equation*}
    Summing over all elements in $E$, we have
    \begin{equation*}
        \tau=2\sum_{i=1}^{n}(2n+1)^{m+n-i}+4\sum_{j=1}^{m}(2n+1)^{m-j}
    \end{equation*}
    Note that it is helpful to view the construction in base $2n+1$, as this makes several of the claims in the proof easier to verify. For example, in base $2n+1$, it becomes immediately apparent that all elements in $E$ are distinct.

    Consider any multiset $E'$ of size $n+1$. We now show that if there exists an index $i\in[n+1]$ such that $E'$ contains zero or multiple occurrences of the elements in $\{t_{i},f_{i}\}$, then the sum of $E'$ is not equal to $\tau/2$. Since any multiset with a duplicated element necessarily violates this condition, it follows that no multiset of size $|E|/2$ with duplicates can sum to $\tau/2$, thereby satisfying the promise condition.

    Let $i^{*}$ be the smallest such index, and for all $i<i^{*}$, let $\ell_{i}\in\{t_{i},f_{i}\}$ denote the unique element from the pair that appears in $E'$. If $E'$ contains multiple occurrences of elements from $\{t_{i^{*}},f_{i^{*}}\}$, let $\ell_{i^{*}}^{1}$ and $\ell_{i^{*}}^{2}$ denote two such occurrences. Then,
    \begin{equation*}
        \sum_{\ell\in E'}\ell
        \geq\ell_{1}+\dots+\ell_{i^{*}-1}+\ell_{i^{*}}^{1}+\ell_{i^{*}}^{2}
        >\frac{\tau}{2},
    \end{equation*}
    where the last inequality holds because
    \begin{align*}
        \ell_{1}+\dots+\ell_{i^{*}-1}+\ell_{i^{*}}^{1}+\ell_{i^{*}}^{2}
        &=(11\cdots 11\underline{2}00\cdots 00\,|\,\tilde{3}\tilde{3}\cdots\tilde{3}\tilde{3})_{2n+1}\\
        &>(11\cdots 11\underline{1}33\cdots 33\,|\,33\cdots 33)_{2n+1}\\
        &>(11\cdots 11\underline{1}11\cdots 11\,|\,22\cdots 22)_{2n+1}\\
        &=\tau/2.
    \end{align*}
    In the expressions above, numbers are written in base $2n+1$. The underlined digit corresponds to the position indexed by $m+n-i^{*}$. The bar separates the representation into two parts: the left portion has $n$ digits, and the right portion has $m$ digits. A tilde over a digit indicates that the digit is at most that value.

    If $E'$ contains no occurrence of elements from $\{t_{i^{*}},f_{i^{*}}\}$, let $\bar{E}'=E'\setminus\{\ell_{1},\dots,\ell_{i^{*}-1}\}$. Then for each $\ell\in\bar{E}'$, we have
    \begin{align*}
        \ell
        &=(00\cdots 00\underline{0}\tilde{1}\tilde{1}\tilde{1}\tilde{1}\cdots \tilde{1}\tilde{1}\,|\,\tilde{3}\tilde{3}\cdots\tilde{3}\tilde{3})_{2n+1}\\
        &\leq(00\cdots 00\underline{0}11\tilde{1}\tilde{1}\cdots \tilde{1}\tilde{1}\,|\,\tilde{3}\tilde{3}\cdots\tilde{3}\tilde{3})_{2n+1}\\
        &<(00\cdots 00\underline{0}1200\cdots 00\,|\,00\cdots 00)_{2n+1}\\
        &=(2n+1)^{m+n-i^{*}-1}+2(2n+1)^{m+n-i^{*}-2}.
    \end{align*}
    Furthermore, observe that
    \vspace{4pt} 
    \begin{align*}
        \smash{\sum_{i=1}^{\mathclap{i^{*}-1}}\ell_{i}+(2n+1)^{n+m-i^{*}}}
        &=(11\cdots 11\underline{1}00\cdots 00\,|\,\tilde{3}\tilde{3}\cdots\tilde{3}\tilde{3})_{2n+1}\\
        &<(11\cdots 11\underline{1}11\cdots 11\,|\,55\cdots 55)_{2n+1}\\
        &<\frac{\tau}{2}+3\sum_{j=1}^{m}(2n+1)^{m-j}.
    \end{align*}
    Combining these results, we obtain
    \begin{equation*}
        \sum_{\ell\in S'}\ell
        <\sum_{i=1}^{\mathclap{i^{*}-1}}\ell_{i}+n(2n+1)^{m+n-i^{*}-1}+2n(2n+1)^{m+n-i^{*}-2}
        <\frac{\tau}{2},
    \end{equation*}
    where the final inequality follows immediately by substituting the earlier bound and observing that
    \begin{equation*}
        n(2n+1)^{m+n-i^{*}-1}+2n(2n+1)^{m+n-i^{*}-2}+3\sum_{j=1}^{m}(2n+1)^{m-j}<(2n+1)^{n+m-i^{*}}.
    \end{equation*}
    Thus, our promise condition is satisfied.

    We now show that the constructed \textsc{PBP} instance $E$ is a \texttt{YES} instance if and only if the \textsc{1-in-3-sat} instance $\Phi$ is a \texttt{YES} instance.

    $(\Rightarrow)$ Suppose $E$ is a \texttt{YES} instance, and let $E_{1},E_{2}$ be a valid partition. Without loss of generality, assume that $t_{n+1}\in E_{1}$. Furthermore, by construction, for each $i\in[n+1]$, the set $E_{1}$ must contain exactly one element from the pair $\{t_{i},f_{i}\}$. We define a truth assignment as follows: for each $i\in[n]$, set $x_{i}=\texttt{TRUE}$ if $t_{i}\in E_{1}$ and $x_{i}=\texttt{FALSE}$ if $f_{i}\in E_{1}$. Then,
    \vspace{1pt} 
    \begin{align*}
        \smash{\sum_{\mathclap{\ell\in E_{1}}}\ell-t_{n+1}}
        &=(11\cdots 11\,|\,22\cdots 22)_{2n+1}-(00\cdots 00\,|\,11\cdots 11)_{2n+1}\\
        &=(11\cdots 11\,|\,11\cdots 11)_{2n+1}.
    \end{align*}
    Here, the least significant $m$ digits (after the vertical bar) count how many literals are satisfied in each clause. Since each digit is exactly $1$, this implies that, under the assignment $x_{1},\dots,x_{n}$, each clause of $\Phi$ is satisfied by exactly one literal. Hence, $\Phi$ is a \texttt{YES} instance.

    $(\Leftarrow)$ Suppose $\Phi$ is a \texttt{YES} instance, and let $x_{1},\dots,x_{n}$ be a satisfying assignment such that exactly one literal is satisfied in each clause. For each $i\in[n]$, let $\ell_{i}=t_{i}$ if $x_{i}=\texttt{TRUE}$ and $\ell_{i}=f_{i}$ if $x_{i}=\texttt{FALSE}$. Additionally, let $\ell_{n+1}=t_{n+1}$, and define the set $E=\{\ell_{1},\dots,\ell_{n+1}\}$. Since exactly one of $\{t_{i},f_{i}\}$ is selected for each $i\in[n]$, each of the $n$ most significant digit receives exactly one contribution. Moreover, because the assignment satisfies each clause exactly once, each of the $m$ least significant digit also receives exactly one contribution. Therefore, the sum of the elements $\ell_{1}+\dots+\ell_{n}=(11\cdots 11\,|\,11\cdots 11)_{2n+1}$. Adding $\ell_{n+1}$ to it, we get
    \vspace{1pt} 
    \begin{align*}
        \smash{\sum_{\mathclap{\ell\in E_{1}}}\ell}
        &=(11\cdots 11\,|\,11\cdots 11)_{2n+1}+(00\cdots 00\,|\,11\cdots 11)_{2n+1}\\
        &=(11\cdots 11\,|\,22\cdots 22)_{2n+1}\\
        &=\tau/2.
    \end{align*}
    We note that $|E_{1}|=|E_{2}|=|E|/2$ by construction. Hence, $E$ is a \texttt{YES} instance.

    We now complete the proof by reducing from the \textsc{PBP} problem. Given a set $E=\{e_{1},\dots,e_{k}\}$ from an instance of \textsc{PBP}, we construct an instance $\calI=\instfull$ of \textsc{ERM} with $n=m=k$, $T=n/2$, $\kappa=\tau/2$, and identical valuations defined by $u_{i}(g_{j})=e_{j}$ for all $i\in[n]$ and $j\in[m]$. We claim that there exists a balanced partition $E_{1},E_{2}$ of $E$ if and only if there exists a sequence of matchings in $\calI$ that achieves a bottleneck value of $\kappa$.

    $(\Rightarrow)$ Suppose $E$ is a \texttt{YES} instance, and let $E_{1},E_{2}$ be a balanced partition. Since there are $T=n/2$ rounds, construct an allocation in which the first $n/2$ agents receive, over the course of the $T$ rounds, all the goods in $E_{1}$, one per round, and the remaining $n/2$ agents receive all the goods in $E_{2}$, again one per round. Under this allocation, each agent receives exactly $T$ goods and accumulates a total value of $\tau/2=\kappa$, and each good is matched once per round and appears in exactly $T$ rounds. By \Cref{lem:alloc-to-seq}, there exists a sequence $S\in\bbS^{T}$ such that $v_{i}^{T}(S)\geq\kappa$ for all $i\in N$. Hence, $(\calI,\kappa)$ is a \texttt{YES} instance.
    
    $(\Leftarrow)$ Suppose $(\calI,\kappa)$ is a \texttt{YES} instance, and let $S\in\bbS^{T}$ be a sequence of matchings such that each agent $i\in N$ receives total value $v_{i}^{T}(S)\geq\kappa$. Since each good appears in exactly $T=n/2$ rounds and all valuations are identical, the total value across all agents is exactly $n\kappa$. But since each agent receives at least $\kappa$, and there are $n$ agents, it follows that each agent must receive exactly $\kappa$. Now consider the multiset of goods that agent $1$ receives under $S$. By the promise condition, any multiset of size $n/2=k/2$ whose sum is $\kappa=\tau/2$ must consist of distinct elements. Therefore, the goods assigned to agent $1$ are all distinct. Let $E_{1}\subset E$ be the set of integers corresponding to the goods received by agent $1$, and let $E_{2}=E\setminus E_{1}$. Then, $|E_{1}|=|E_{2}|=k/2$ and the sum of each set is $\tau/2$. Thus, $E_{1}$ and $E_{2}$ form a balanced partition. Hence, $E$ is a \texttt{YES} instance.
\end{proof}

\optimalidenticalfactor*

\begin{proof}
    Let $G_{*}\subseteq G$ be the top $n$ most valuable goods and consider the allocation $A$ that gives $k$ copies of each good in $G_{*}$ to each agent $i\in N$. Then, we have $v_{i}(A)=v_{i'}(A)$ for all agents $i,i'\in N$. Suppose, for sake of contradiction, that $A$ is not optimal. Then, there exist some other allocation $A'$ such that $\min_{i}v_{i}(A')>\min_{i}v_{i}(A)$. This implies that
    \begin{equation*}
        nk\sum_{\mathclap{g\in G_{*}}}u_{i}(g)
        \geq \sum_{i\in N}v_{i}(A')
        \geq n\cdot\min_{i\in N}v_{i}(A')  >n\cdot\min_{i\in N}v_{i}(A)
        =\sum_{i\in N}v_{i}(A)
        =nk\sum_{\mathclap{g\in G_{*}}}u_{i}(g),
    \end{equation*}
    where the first inequality is true because there is no way to achieve strictly greater utilitarian value than by assigning out the top $n$ most valuable goods in every round. Since this leads to a contradiction, we conclude that $A$ is optimal.
\end{proof}

{\renewcommand\footnote[1]{}\approxidentical*}

\begin{proof}
    We first describe the polynomial-time algorithm that will return us a sequence of matchings $S$. At each round $t\in[T]$, sort the agents in increasing order of cumulative valuation up till round $t-1$. Then, in this order, let each agent choose their favorite good and allocate it to them. Repeat this process until all rounds are completed. Note that since we are considering the setting with identical valuations, it suffices to only look at the top $n$-valued goods---no agent will choose any of the other (lower-valued) goods in any round.
    
    Fix any round $t\in[T]$. It is easy to observe that
    \begin{equation}
        \label{eqn:ident-opt}
        \frac{1}{n}\cdot\sum_{i\in N} v_{i}^{t}(S)\geq\OPT(t),\tag{1}
    \end{equation}
    since our algorithm allows agents to select their favorite good in increasing order of cumulative value up till round $t$. Let the bottleneck agent at round $t$ be $i$, that is, $b^{t}(S)=v_{i}^{t}(S)$. Let $N'$ be the set of agents who picks a good before agent $i$ at some point. Then, for each agent $i'\in N'$, let $s_{0}\leq t$ be the last round in which $i'$ picks a good before $i$. By the algorithm, we have
    \begin{equation*}
        v_{i'}^{s_{0}-1}(S)\leq v_{i}^{s_{0}-1}(S)
        \quad\text{and}\quad
        u(M^{s}(i'))\leq u(M^{s}(i))\text{ for all rounds $s\in[s_{0}+1,t]$}.
    \end{equation*}
    By the first inequality, we have
    \begin{align*}
        v_{i'}^{s_{0}}(S)
        &=v_{i'}^{s_{0}-1}(S)+u(M^{s_{0}}(i'))\\
        &\leq v_{i}^{s_{0}-1}(S)+u(M^{s_{0}}(i))-u(M^{s_{0}}(i))+u(M^{s_{0}}(i'))\\
        &=v_{i}^{s_{0}}(S)-u(M^{s_{0}}(i))+u(M^{s_{0}}(i'))\\
        &\leq v_{i}^{s_{0}}(S)+\Delta.
    \end{align*}
    Combining this result with the second inequality, we have
    \begin{equation*}
        v_{i'}^{t}(S)
        =v_{i'}^{s_{0}}(S)+\sum_{\mathclap{s=s_{0}+1}}^{t} u(M^{s}(i'))
        \leq v_{i}^{s_{0}}(S)+\Delta+\sum_{\mathclap{s=s_{0}+1}}^{t} u(M^{s}(i))
        =v_{i}^{t}(S)+\Delta.
    \end{equation*}
    Furthermore, for each $i'\in N\setminus N'$, we have
    \begin{equation*}
        v_{i'}^{t}(S)
        =\sum_{\mathclap{s=1}}^{t} u(M^{s}(i'))
        \leq \sum_{\mathclap{s=1}}^{t} u(M^{s}(i))
        =v_{i}^{t}(S)
        \leq v_{i}^{t}(S)+\Delta.
    \end{equation*}
    Taking the average over all agents and using (\ref{eqn:ident-opt}), we get
    \begin{equation*}
        \OPT(t)
        \leq\frac{1}{n}\cdot\sum_{\mathclap{i'\in N}}v_{i'}^{t}(S)
        \leq v_{i}^{t}(S)+\Delta
        =b^{t}(S)+\Delta
    \end{equation*}
    as desired.
\end{proof}